\documentclass{article}
\usepackage{kr}

\usepackage{times}
\usepackage{latexsym}
\usepackage{soul}
\usepackage{url}
\usepackage[hidelinks]{hyperref}
\usepackage[utf8]{inputenc}
\usepackage[small]{caption}
\usepackage{graphicx}
\usepackage{amsmath}
\usepackage{soul}
\usepackage{amsthm}
\usepackage{booktabs}
\usepackage{algorithm}
\usepackage{algorithmic}
\usepackage[all]{xy}
\usepackage{comment}
\usepackage{amssymb}
\usepackage{xcolor}
\usepackage{pifont}
\usepackage{ebproof}
\usepackage{xspace}
\usepackage{array}
\usepackage{tikz}
\usepackage{lineno}[switch]
\usepackage{multibib}

\newcites{ap}{References}

\newcommand{\agata}[1]{{\color{blue}#1}}

\newcommand{\yesmark}{\checkmark}
\newcommand{\nomark}{ }

\newcommand{\figureword}{Fig.}

\newcommand{\lemmaword}{Lem.}
\newcommand{\propword}{Prop.}
\newcommand{\theoremword}{Th.}
\newcommand{\propositionword}{Prop.}

\newcommand{\tableword}{Tab.}
\newcommand{\appendixword}{Appendix}

\newcommand{\seq}{\; \vdash \;}
\newcommand{\seqcl}{\, \models \,}
\newcommand{\seqlk}{\Rightarrow}
\newcommand{\seqind}[1]{\; \vdash_{#1} \;}

\newcommand{\scomma}{, \;}

\newcommand{\inputworldset}{\textit{In}}
\newcommand{\inputworld}{\textit{in}}
\newcommand{\outputworld}{\textit{out}}
\newcommand{\otvalid}{1-2-valid\xspace}
\newcommand{\otvalidity}{1-2-validity\xspace}
\newcommand{\tfvalid}{3-4-valid\xspace}
\newcommand{\tfvalidity}{3-4-validity\xspace}

\newcommand{\ioseq}[3]{{#1 \vdash (#2, #3)}}
\newcommand{\rulename}[1]{\textcolor{black}{\text{(#1)}}}
\newcommand{\inaxlabel}{IN}
\newcommand{\inaxplabel}{IN-AX$^+$}
\newcommand{\outaxlabel}{OUT}
\newcommand{\outaxplabel}{OUT-AX$^+$}
\newcommand{\pairelimlabel}[1]{E$_#1$}
\newcommand{\weaklabel}{IO-Wk}
\newcommand{\contractlabel}{IO-Ctr}
\newcommand{\cutlabel}{IO-Cut}
\newcommand{\weaklabelLK}{Wk}
\newcommand{\contractlabelLK}{Ctr}
\newcommand{\cutlabelLK}{Cut}
\newcommand{\Krulelabel}{$\square$ R}

\newcommand{\iologic}[1]{\text{OUT$_#1$}}
\newcommand{\ioslogic}[1]{\text{OUT$^{\bot}_#1$}}

\newcommand{\iocalc}[1]{\text{${\mathcal SC}^\bot_#1$}}

\newcommand{\partset}{\mathcal{P}}
\newcommand{\varset}{\mathcal{V}}
\newcommand{\varbij}{\nu}
\newcommand{\varlab}[2]{#1^{#2}}

\newcommand{\propencoding}[3]{{\mathcal{P}^{#1}_{#2}(#3)}}

\newcommand{\numofinputworlds}[1]{\mathcal{N}_{#1}}

\newtheorem{remark}{Remark}
\newtheorem{theorem}{Theorem}
\newtheorem{proposition}{Proposition}
\newtheorem{lemma}{Lemma}
\newtheorem{definition}{Definition}
\newtheorem{notation}{Notation}
\newtheorem{corollary}{Corollary}

\newtheoremstyle{TheoremNum}
    {\topsep}{\topsep}              
    {\itshape}                      
    {}                              
    {\bfseries}                     
    {.}                             
    { }                             
    {\thmname{#1}\thmnote{ \bfseries #3}}
\theoremstyle{TheoremNum}
\newtheorem{lemman}{Lemma}
\newtheorem{thn}{Theorem}

\newcolumntype{U}{>{\centering\arraybackslash}p{0.045\textwidth}}
\newcolumntype{V}{>{\centering\arraybackslash}p{0.035\textwidth}}

\usetikzlibrary{arrows,calc,patterns,positioning,shapes}
\usetikzlibrary{decorations.pathmorphing}
\tikzset{
modal/.style={>=stealth',shorten >=1pt,shorten <=1pt,auto,
node distance=1.5cm,semithick},
world/.style={circle,draw,minimum size=1cm,fill=gray!15},
point/.style={circle,draw,fill=black,inner sep=0.5mm},
reflexive/.style={->,in=120,out=60,loop,looseness=#1},
reflexive/.default={5},
reflexive point/.style={->,in=135,out=45,loop,looseness=#1},
reflexive point/.default={25},
reflexive above/.style={->,loop,in=120,out=60,looseness=#1},
reflexive above/.default={7},
reflexive below/.style={->,loop,in=240,out=300,looseness=#1},
reflexive below/.default={7},
reflexive left/.style={->,loop,in=150,out=210,looseness=#1},
reflexive left/.default={7},
reflexive right/.style={->,loop,in=30,out=330,looseness=#1},
reflexive right/.default={7}
}

\title{Streamlining Input/Output Logics with Sequent Calculi --- Extended Version}

\author{%
Agata Ciabattoni\and
Dmitry Rozplokhas \\
\affiliations
Institute of Logic and Computation, TU Wien, Vienna, Austria \\
\emails
\{agata, dmitry\}@logic.at
}

\begin{document}



\maketitle

\begin{abstract}
Input/Output (I/O) logic is a general framework for reasoning about conditional norms and/or causal relations. 
We streamline Bochman's causal I/O logics via proof-search-oriented sequent calculi. 
 Our calculi establish a natural syntactic link between the derivability in these logics and in the original I/O logics. 
 As a consequence of our results, we obtain new, simple semantics for all these logics, complexity bounds, embeddings into normal modal logics, and efficient deduction methods. 
Our work encompasses many scattered results and provides uniform solutions to various unresolved problems.

\end{abstract}

\section{Introduction}

Input/Output (I/O) logic is a general framework proposed by~\cite{IO_original}
to reason about conditional norms. 
I/O logic is not a single logic but rather a family of logics, each viewed as a “transformation engine”, which converts an input (condition under which the obligation holds) into an output (what is obligatory under these conditions).
Many different I/O logics have been defined,
e.g.,~\cite{IO_original_constrained,IOhandbook,IO_minimal_aggregative,IO_conceptual}, and also used as building blocks for causal reasoning~\cite{causal_logic,causality_for_nonmonotonic,Perls_causality_logic,causality_book},
laying down the logical foundations for the causal calculus~\cite{McCainT97}, and for legal reasoning~\cite{kelsenian_logic}. 
I/O logics manipulate Input-Output pairs\footnote{Production rules $A \seqlk B$, in Bochman's terminology.} $(A, B)$, which consists of boolean formulae representing either conditional obligations (in the case of the original I/O logics) or causal relations ($A$ causes $B$, in the case of their causal counterparts).
Different I/O logics are defined by varying the mechanisms of obtaining new pairs from a set of pairs (entailment problem). 
Each I/O logic is characterized by its own semantics. The original I/O logics use a procedural approach, while their causal counterparts adopt bimodels, which in general consist of pairs of arbitrary deductively-closed sets of formulae. Additionally, each I/O logic is equipped with a proof calculus, consisting of axioms and rules but not suitable for proof search.

This paper deals with the four original I/O logics \iologic{1}-\iologic{4} in \cite{IO_original} and their
causal counterpart \ioslogic{1}-\ioslogic{4} in~\cite{causality_for_nonmonotonic}.
We introduce proof-search-oriented sequent calculi and use them to bring together scattered results and to provide uniform solutions to various unresolved problems.
%
%
%
%
Indeed~\cite{IO_sequents_argumentation}
characterized 
many I/O logics through an argumentative approach using sequent-style calculi. Their calculi are not proof search-oriented. First sequent calculi of this kind for some I/O logics, including \iologic{1} and \iologic{3}, 
have been proposed in~\cite{IO_sequents}.
Their implementation provides an alternative decidability proof, although not optimal (entailment is shown to be in 
$\Pi^P_3$).
Moreover, the problem of finding proof-search-oriented calculi for \iologic{2} and \iologic{4} was left open there. A prover for these two logics was introduced in~\cite{IO_in_HOL}. The prover encodes in classical Higher Order Logic their embeddings from~\cite{IO_original} into the normal modal logics {\bf K} and {\bf KT}.
%
Finding an embedding of \iologic{1} and \iologic{3} into normal modal logics was left as an open problem, that~\cite{IOhandbook} indicates as difficult, if possible at all.
An encoding of various I/O logics into more complicated logics (adaptive modal logics) is in~\cite{IO_adaptive_characterization}. 
Using their procedural semantics, \cite{IO_decision_procedures} defined
goal-directed decision procedures for the
original I/O logics, 
without mentioning the complexity of the task.
\cite{IO_complexity} showed that the entailment problem for \iologic{1}, \iologic{2}, and \iologic{4} is co-NP-complete, while for \iologic{3} the complexity was determined to lie within the first and second levels of the polynomial hierarchy, without exact resolution.
%

In this paper we follow a new path that streamlines the considered logics. Inspired by the modal embedding of \ioslogic{2} and \ioslogic{4} 
in~\cite{causal_logic}, we design well-behaving sequent calculi for Bochman's causal I/O logics.
%
The normal form of derivations in these calculi allows a simple syntactic link between derivability in the original I/O logics and in their causal versions to be established, making it possible to utilize our calculi for the original I/O logics as well.  
As a by-product:
\vspace{-0.1cm}
\begin{itemize}
\item We introduce a simple possible worlds semantics.
\vspace{-0.1cm}
\item We prove co-NP-completeness and provide 
efficient automated procedures for the entailment problem; the latter are obtained
via reduction to unsatisfiability of a classical logic formula of polynomial size.
\vspace{-0.1cm}
\item We provide embeddings into the shallow fragment of the modal logics $\mathbf{K}$, $\mathbf{KD}$ (i.e., standard deontic logic~\cite{SDL}), and their extension with axiom $\mathbf{F}$.
\end{itemize}
These results are uniformly obtained 
{\em for all} four original I/O logics and their causal versions.

\section{Preliminaries}
In the I/O logic framework, conditional norms (or causal relations) are expressed as pairs $(B, Y)$ of propositional boolean formulae.   
The semantics is operational, rather than truth-functional.
The meaning of the deontic/causal concepts in these logics is given in terms of a set of procedures yielding outputs for inputs. The basic mechanism underpinning these procedures is detachment (modus ponens).

On the syntactic side, different I/O logics are obtained by varying the mechanisms of obtaining new input-output pairs from a given set of these pairs. The mechanisms
introduced in the original paper~\cite{IO_original} are based on the following (axioms and) rules ($\models$ denotes semantic entailment in classical propositional logic):

%
 
\begin{description}
\item[\rulename{TOP}] $(\top, \top)$ is derivable from no premises
\item[\rulename{BOT}] $(\bot, \bot)$ is derivable from no premises
\item[\rulename{WO}] $(A,X)$ derives $(A, Y)$ whenever $X \seqcl Y$ 
\item[\rulename{SI}] $(A,X)$ derives $(B, X)$ whenever $B \seqcl A$ 
\item[\rulename{AND}] $(A,X_1)$ and $(A,X_2)$ derive $(A, X_1 \wedge X_2)$
\item[\rulename{OR}] $(A_1,X)$ and $(A_2,X)$ derive $(A_1 \vee A_2, X)$
\item[\rulename{CT}] $(A,X)$ and $(A \wedge X,Y)$ derive $(A, Y)$
\end{description}

Different I/O logics are given by different subsets $R$ of these rules, see \figureword~\ref{tab:io_logics_rule_schemes}. The basic system, called {\em simple-minded output} \iologic{1},  consists of the rules $\{\rulename{TOP},$ $\rulename{WO},$ $ \rulename{SI},$ $ \rulename{AND}\}$.
Its extension with \rulename{OR} (for reasoning by cases) leads to {\em basic output} logic \iologic{2}, with \rulename{CT} (for reusability of outputs as inputs in derivations) to {\em simple-minded reusable output} logic \iologic{3}, and with both \rulename{OR} and \rulename{CT} to {\em basic reusable output} logic \iologic{4}. 
Their causal counterpart~\cite{causality_for_nonmonotonic}, 
that we denote by \ioslogic{i} for $i= 1, \dots , 4$, extends the corresponding logics with \rulename{BOT}. 
\begin{definition}
\label{def:IO_derivability}
Given a set of pairs $G$ and
a set $R$ of rules,
a derivation in an I/O logic of a pair $(B, Y)$ from $G$ is 
a tree with $(B, Y)$ at the root, each non-leaf node derivable from its immediate parents by one of the rules in $R$, and each leaf node is an element of $G$ or an axiom from $R$.
\end{definition}

\begin{table}
\centering
\scalebox{0.85}{\begin{tabular}
{|c| c  c  c  c  c  c  c|}
\hline
Logic  &
\rulename{TOP} &
\rulename{BOT} &
\rulename{WO} &
\rulename{SI} &
\rulename{AND} &
\rulename{OR} &
\rulename{CT} \\[1pt]
\hline
\iologic{1} &
\yesmark &
\nomark &
\yesmark &
\yesmark &
\yesmark &
\nomark &
\nomark \\[1pt]
\iologic{2} &
\yesmark &
\nomark &
\yesmark &
\yesmark &
\yesmark &
\yesmark &
\nomark \\[1pt]
\iologic{3} &
\yesmark &
\nomark &
\yesmark &
\yesmark &
\yesmark &
\nomark &
\yesmark \\[1pt]
\iologic{4} &
\yesmark &
\nomark &
\yesmark &
\yesmark &
\yesmark &
\yesmark &
\yesmark \\[1pt]
\ioslogic{1} &
\yesmark &
\yesmark &
\yesmark &
\yesmark &
\yesmark &
\nomark &
\nomark \\[1pt]
\ioslogic{2} &
\yesmark &
\yesmark &
\yesmark &
\yesmark &
\yesmark &
\yesmark &
\nomark \\[1pt]
\ioslogic{3} &
\yesmark &
\yesmark &
\yesmark &
\yesmark &
\yesmark &
\nomark &
\yesmark \\[1pt]
\ioslogic{4} &
\yesmark &
\yesmark &
\yesmark &
\yesmark &
\yesmark &
\yesmark &
\yesmark \\[1pt]
\hline
\end{tabular}}

\setlength{\belowcaptionskip}{-10pt}
\caption{Defining rules for the considered I/O logics}
\label{tab:io_logics_rule_schemes}
\end{table}
We indicate by $G \vdash_{OUT\ast} (B,Y)$
that the pair $(B, Y)$ is derivable in the I/O logic $OUT\ast$
from the set of pairs in $G$ ({\em entailment problem}). 
We will refer to $(B, Y)$ as the \emph{goal pair}, to the formulae $B$ and $Y$ as the \emph{goal input} and \emph{goal output} respectively, and to the pairs in $G$ as \emph{deriving pairs}.

\begin{remark}
A derivation in I/O logic is a sort of natural deduction proof, acting on pairs, rather than formulae. This proof theory is however not helpful to decide whether 
$G \vdash_{OUT\ast} (B,Y)$ holds, or to prove metalogical results (e.g., complexity bounds). The main reason is that derivations have no well-behaved normal forms, and in general are difficult to find. 
\cite{IO_sequents} introduced the first proof-search oriented calculi, which operate effectively only in the absence of \rulename{OR} (hence not  
 for \ioslogic{2} and \ioslogic{4}).
The calculi use sequents 
that manipulate pairs expressed using the conditional logic connective $>$.
\end{remark}


\section{Sequent Calculi for Causal I/O Logics}
\label{sec:sequent}
We present sequent-style calculi for the causal I/O logic \ioslogic{1}-\ioslogic{4}.
Their basic objects are 
$$\mbox{{\em I/O sequents}} \quad \ioseq{(A_1, X_1), \dots, (A_n, X_n)}{B}{Y}$$ dealing with pairs, as well as $$\mbox{{\em Genzen's LK sequents}} \quad A_1, \dots A_n \seqlk B_1, \dots , B_m$$ 
dealing with boolean formulae (meaning that $\{A_1, \dots, A_n\} $ $ \models (B_1 \vee \dots \vee B_m)$).
Our calculi are defined by extending the sequent calculus LK\footnote{We assume the readers to be familiar with LK
(a brief overview is given in \appendixword~\ref{sec:ap_LK})
.}
for classical logic 
%
with three rules manipulating I/O sequents: one {\em elimination rule} ---different for each logic--- that removes one of the deriving pair while modifying the goal pair, and two {\em concluding rules} that transform the derivation of the goal pair into an LK derivation of either the goal input or the goal output. The latter rules, which are the same for all the considered logics, are in \figureword~\ref{fig:IO_concluding_rules}.

\begin{figure}
    \centering
    \begin{tabular}{cc}
\scalebox{0.90}{\begin{prooftree}
    \hypo{B \seqlk}
    \infer1[\rulename{\inaxlabel}]{\ioseq{G}{B}{Y}}
\end{prooftree}} \quad & \quad
\scalebox{0.90}{\begin{prooftree}
    \hypo{\seqlk Y}
    \infer1[\rulename{\outaxlabel}]{\ioseq{G}{B}{Y}}
\end{prooftree}}
    \end{tabular}
    \setlength{\belowcaptionskip}{-12pt}
    \caption{Concluding rules (same for all causal I/O logic)}
    \label{fig:IO_concluding_rules}
\end{figure}

\begin{definition}
A {\em derivation} in our calculi is a finite labeled tree whose internal nodes are I/O or $LK$ sequents s.t. the label of each node follows from the labels of
its children using the calculus rules.
We say that an I/O sequent $\ioseq{(A_1, X_1), \dots, (A_n, X_n)}{B}{Y}$ is derivable if all the leaves of its derivation are LK axioms. 
\end{definition}


A derivation of an I/O sequent
consists of two phases. Looking at it bottom up, we first encounter rules dealing with pairs (pair elimination and concluding rules) followed by LK rules.
The calculi, in a sense, uphold the ideological principles guiding I/O logics: pairs (i.e., conditional norms) are treated separately from the boolean statements.

It is easy to see that using (\inaxlabel) and (\outaxlabel) 
we can derive 
$\rulename{TOP}$ and $\rulename{BOT}$; their soundness 
in the weakest causal I/O logic $\ioslogic{1}$ is proven below.
\begin{lemma}
\label{lemma:concluding}
(\inaxlabel) and (\outaxlabel) are derivable in 
$\ioslogic{1}$.
\end{lemma}
\begin{proof}
If $B \seqlk$  we have the following  derivation in $\ioslogic{1}$: from
$(\bot,\bot)$ and $B \seqcl \bot$ (i.e., $B \seqlk$) we get $(B, \bot)$ by $\rulename{SI}$; the required pair $(B,Y)$ follows by $\rulename{WO}$.

Assume $\seqlk Y$. From  
$(\top,\top)$ and $\seqlk Y$ (i.e., $\top \seqcl Y$) by $\rulename{WO}$ we get $(\top, Y)$, from which $(B,Y)$ follows by $\rulename{SI}$.
\end{proof}

\vspace{-2mm}

Henceforth, when presenting derivations in our calculi, we will omit the
LK sub-derivations.
\subsection{Basic Production Inference \ioslogic{2}}
\label{sec:sequent_2}
The calculus \iocalc{2} for the causal\footnote{Called basic production inference in~\cite{causality_for_nonmonotonic}.} basic output logic \ioslogic{2} 
is obtained by adding to the core calculus (consisting of $LK$ with the rules $(\inaxlabel)$ and $(\outaxlabel)$) the pair elimination rule \rulename{\pairelimlabel{2}} in 
\figureword~\ref{fig:IO_pair_elimination_rules}.

\begin{remark}
The rule  \rulename{\pairelimlabel{2}} is inspired by
the embedding in~\cite{causal_logic}
of \ioslogic{2} into the modal logic $\mathbf{K}$: ${(A_1, X_1) \scomma \dots \scomma (A_n, X_n) \vdash_{\ioslogic{2}} (B, Y)}$  iff $ (\ast) \;  {(A_1 \to \square X_1) \scomma \dots \scomma (A_n \to \square X_n) \scomma B \seq \square Y}$ is derivable in $\mathbf{K}$. To provide the rule's intuition we make use of
the sequent calculus $GK$ for $\mathbf{K}$ in~\cite{GK_calculus}. $GK$ 
extends LK with the following rule for introducing boxes (or eliminating them, looking at the rule bottom up):
\[ \scalebox{0.90}{\begin{prooftree}
    \hypo{A_1 \scomma \dots \scomma A_n \seq B}
    \infer1[\rulename{\Krulelabel}]{\square A_1, \dots, \square A_n \seq \square B}
\end{prooftree}}  \]
%
To prove the sequent $(\ast)$ in $GK$ we can apply the $LK$ rule for $\to$ to one of the implications $(A_i \to \square X_i)$ on the left. This creates two premises: (a) ${G' \scomma B \seqind{K} \square Y \scomma A_i}$ and (b) ${G' \scomma B \scomma \square X_i \seqind{K} \square Y}$ (where $G'$ is the set of all implications on the left-hand side but ${(A_i \to \square X_i)}$). Now (a) ${G' \scomma B \seqind{K} \square Y \scomma A_i}$ is equivalent in $\mathbf{K}$ to (derivable in $GK$ if and only if so is) the sequent ${G' \scomma (B \wedge \neg A_i) \seqind{K} \square Y}$, that using the embedding again leads to the first premise of \rulename{\pairelimlabel{2}}; (b)
${G' \scomma B \scomma \square X_i \seqind{K} \square Y}$ is equivalent (for suitable $G'$) to ${G' \scomma B \seqind{K} \square (Y \vee \neg X_i)}$, 
which leads to the second premise of 
\rulename{\pairelimlabel{2}}.
%
\end{remark}

\begin{figure*}
    \centering
    \begin{tabular}{cc}
\scalebox{0.90}{\begin{prooftree}
    \hypo{\ioseq{G}{B \land \neg A}{Y}}
    \hypo{\ioseq{G}{B}{Y \lor \neg X}}
    \infer2[\rulename{\pairelimlabel{2}}]{\ioseq{(A, X) \scomma G}{B}{Y}}
\end{prooftree}} &
\scalebox{0.90}{\begin{prooftree}
    \hypo{B \seqlk A}
    \hypo{\ioseq{G}{B}{Y \lor \neg X}}
    \infer2[\rulename{\pairelimlabel{1}}]{\ioseq{(A, X) \scomma G}{B}{Y}}
\end{prooftree}} \\[5mm]
\scalebox{0.90}{\begin{prooftree}
    \hypo{\ioseq{G}{B \land \neg A}{Y}}
    \hypo{\ioseq{G}{B \land X}{Y \lor \neg X}}
    \infer2[\rulename{\pairelimlabel{4}}]{\ioseq{(A, X) \scomma G}{B}{Y}}
\end{prooftree}} &
\scalebox{0.90}{\begin{prooftree}
    \hypo{B \seqlk A}
    \hypo{\ioseq{G}{B \land X}{Y \lor \neg X}}
    \infer2[\rulename{\pairelimlabel{3}}]{\ioseq{(A, X) \scomma G}{B}{Y}}
\end{prooftree}}
    \end{tabular}
    \caption{Sequent rules for pair elimination (one for each considered causal I/O logic)}
    \label{fig:IO_pair_elimination_rules}
\end{figure*}


\begin{figure*}
    \centering
    \begin{tabular}{ccc}
\scalebox{0.90}{\begin{prooftree}
    \hypo{\ioseq{G}{B}{Y}}
    \infer1[\rulename{\weaklabel}]{\ioseq{(A, X), G}{B}{Y}}
\end{prooftree}} &
\scalebox{0.90}{\begin{prooftree}
    \hypo{\ioseq{(A, X), (A, X), G}{B}{Y}}
    \infer1[\rulename{\contractlabel}]{\ioseq{(A, X), G}{B}{Y}}
\end{prooftree}} &
\scalebox{0.90}{\begin{prooftree}
    \hypo{\ioseq{G}{C}{Z}}
    \hypo{\ioseq{(C, Z), G'}{B}{Y}}
    \infer2[\rulename{\cutlabel}]{\ioseq{G, G'}{B}{Y}}
\end{prooftree}}
    \end{tabular}
    \setlength{\belowcaptionskip}{-10pt}
    \caption{Structural I/O rules (admissible in all our calculi)}
    \label{fig:IO_structural_rules}
\end{figure*}

We prove below the soundness and completeness of the calculus \iocalc{2} for $\ioslogic{2}$. We start by describing a useful characterization of derivability in \iocalc{2} of an I/O sequent $\ioseq{(A_1, X_1), \dots, (A_n, X_n)}{B}{Y}$ via derivability of certain sequents in LK.


\begin{notation}
$\partset(X)$ will denote the set of all \textit{partitions} of the set $X$, i.e.,
$\partset(X) = \{ (I,J) \mid I \cup J = X, I \cap J = \emptyset \}$
\end{notation}

Notice that 
if a concluding rule $\rulename{\inaxlabel}$ or $\rulename{\outaxlabel}$ can be applied to the conclusion of \rulename{\pairelimlabel{2}}, it can also be applied to its premises. 
This observation implies that if $(A_1, X_1),$ $ \dots, $ $ (A_n, X_n) \vdash (B,Y)$ is derivable in \iocalc{2} there is a derivation in which the concluding rules are applied only when all deriving pairs are eliminated. We use this {\em I/O normal form} of derivations in the proof of the following lemma.

\begin{lemma}[Characterization lemma for \iocalc{2}]
\label{lem:char_2}
$\ioseq{(A_1, X_1), \dots, (A_n, X_n)}{B}{Y}$ is derivable in \iocalc{2} iff for all partitions 
$(I, J) \in \partset(\{1, \dots ,n \})$, either 
${B \seqlk \{A_i\}_{i \in I}}$ or ${\{X_j\}_{j \in J} \seqlk Y}$ is derivable in LK.
\end{lemma}
\begin{proof}
By induction on $n$. 
Base case: $n = 0$. The only partition is $(\emptyset, \emptyset)$.  From derivability of $B \seqlk $ or $\seqlk Y$ follows $\ioseq{}{B}{Y}$ by either $\rulename{\inaxlabel}$ or $\rulename{\outaxlabel}$; the converse also holds.

Inductive case: from $n$ to $n+1$. 
Let $G = \{(A_1, X_1), \dots,$ $ (A_{n+1}, X_{n+1})\}$ and consider only derivations in \iocalc{2} in I/O normal form (so the last applied rule can only be \rulename{\pairelimlabel{2}}).
We characterize the condition when there exists a
derivation of $\ioseq{G}{B}{Y}$ whose last rule applied is \rulename{\pairelimlabel{2}} eliminating a pair $(A_k,X_k)$, for some $k \in {\{1, \dots ,n+1\}}$. 
This application leads to two premises: $\ioseq{G'}{B \wedge \neg A_k}{Y}$ and $\ioseq{G'}{B}{Y \vee \neg X_k}$, where $G' = G \setminus \{(A_k,X_k)\}$. By the inductive hypothesis, the derivability of these premises is equivalent to the derivability of the following sequents: for each $(I', J') \in \partset(\{1, \dots, n+1\} \setminus \{k\})$:
\begin{itemize}
    \item (a1) $B \wedge \neg A_k \seqlk \{A_i\}_{i \in I'}$ or  (a2) $\{X_j\}_{j \in J'} \seqlk Y $, and
    \item (b1) $ B \seqlk \{A_i\}_{i \in I'}$ or (b2) $\{X_j\}_{j \in J'} \seqlk Y \vee \neg X_k$.
\end{itemize}

\noindent
(a1) is equivalent in $LK$ to (a1)' $B \seqlk \{A_i\}_{i \in I' \cup \{k\}}$, and (b2) to (b2)' $\{X_j\}_{j \in J' \cup \{k\}} \seqlk Y$. Hence (a1)', (a2) give the required condition for the partition $(I' \cup \{k\}, J')$, while (b1), (b2)' for the partition $(I', J' \cup \{k\})$. 
\end{proof}

\vspace{-2mm}
The soundness and completeness proof of $\iocalc{2}$ makes use of 
the admissibility in the calculus of the structural rules for I/O sequents (weakening, contraction, and cut)  in 
\figureword~\ref{fig:IO_structural_rules}. Recall that a rule is {\em admissible} if its addition does not change the set of sequents that can be derived.

\begin{lemma}
\label{lem:struct_admissibility_2}
The rules \rulename{\weaklabel}, \rulename{\contractlabel} and \rulename{\cutlabel} in 
\figureword~\ref{fig:IO_structural_rules} are admissible in $\iocalc{2}$.
\end{lemma}
\begin{proof}
By Lem.~\ref{lem:char_2} we can reduce the admissibility of these structural rules to the admissibility of weakening, contraction, and cut in $LK$. Consider the case \rulename{\cutlabel}. Let $G = {\{(D_1, W_1) \dots (D_m, W_m) \}}$, $G' = \{(A_1, X_1), $ $ \dots $ $ (A_n, X_n) \}$ and $(A_{n+1}, X_{n+1}) = (C, Z)$. Now $\ioseq{G, G'}{B}{Y}$ is derivable in \iocalc{2} iff for any  $(I_1, J_1) \in {\partset(\{1,\dots,n\})}$ and $(I_2, J_2) \in {\partset(\{1,\dots,m\})}$ either ${B \seqlk \{A_i\}_{i \in I_1}, \{D_i\}_{i \in I_2}}$  or ${\{X_j\}_{j \in J_1}, \{W_j\}_{j \in J_2} \seqlk Y}$ is derivable in LK.
It is tedious but easy to see that this holds by applying Lem.~\ref{lem:char_2} to the hypotheses
$\ioseq{(C,Z), G}{B}{Y}$ and $\ioseq{G'}{C}{Z}$, and using the structural rules of LK.
(Full proof in \appendixword~\ref{sec:ap_proofs}).
\end{proof}

\begin{theorem}[Soundness and completeness of \iocalc{2}]
\label{th:sound_comp_2}
${G \vdash (B,Y)}$ 
is derivable in \iocalc{2} 
iff $(B, Y)$ is derivable from the pairs in $G$ in \ioslogic{2}.
\end{theorem}
\begin{proof}
{\em (Completeness)} 
Assume that $(B,Y)$ is derivable in \ioslogic{2}.
We prove by induction on the derivation tree that for each pair $(A, X)$ occurring in it, the I/O sequent $\ioseq{G}{A}{X}$ is derivable in \iocalc{2}. The case 
$(A, X) \in G$ is: 
\[ \scalebox{0.90}{\begin{prooftree}
    \hypo{A \wedge \neg A \seqlk}
    \infer1[\rulename{\inaxlabel}]{\ioseq{G'}{A \wedge \neg A}{X}}
    \hypo{\seqlk X \vee \neg X}
    \infer1[\rulename{\outaxlabel}]{\ioseq{G'}{A}{X \vee \neg X}}
    \infer2[\rulename{\pairelimlabel{2}}]{\ioseq{(A,X), G'}{A}{X}}
\end{prooftree}} \]
We show the case of \rulename{SI} ($B \models A$ iff $B \wedge \neg A \seqlk$):
\[ \scalebox{0.90}{\begin{prooftree}
    \hypo{\substack{\textit{by I.H.} \\ \vdots}}
    \infer1{\ioseq{G}{A}{Y}}
    \hypo{B \wedge \neg A \seqlk}
    \infer1[\rulename{\inaxlabel}]{\ioseq{}{B \wedge \neg A}{Y}}
    \hypo{\seqlk Y \vee \neg Y}
    \infer1[\rulename{\outaxlabel}]{\ioseq{}{B}{Y \vee \neg Y}}
    \infer2[\rulename{\pairelimlabel{2}}]{\ioseq{(A,Y)}{B}{Y}}
    \infer2[\rulename{\cutlabel}]{\ioseq{G}{B}{Y}}
\end{prooftree}}  \]
For \rulename{AND}: we derive ${(B, X_1), (B, X_2) \vdash (B, X_1 \land X_2)}$, and apply \rulename{\cutlabel} twice followed by many applications of contraction to the resulting derivation as follows
\[ \scalebox{0.85}{\begin{prooftree}
    \hypo{\substack{\textit{by I.H.} \\ \vdots}}
    \infer1{\ioseq{G}{B}{X_1}}
    \hypo{\substack{\textit{by I.H.} \\ \vdots}}
    \infer1{\ioseq{G}{B}{X_2}}
    \hypo{\ioseq{(B, X_1) \scomma (B, X_2)}{B}{X_1 \land X_2}}
    \infer2
    {\ioseq{G \scomma (B, X_1)}{B}{X_1 \land X_2}}
    \infer2
    {\ioseq{G, G}{B}{X_1 \land X_2}}
    \infer1[\rulename{\contractlabel} $\times n$]{\ioseq{G}{B}{X_1 \land X_2}}
\end{prooftree}}  \]

The claim follows by Lem~\ref{lem:struct_admissibility_2}.
(Full proof in Appendix~\ref{sec:ap_proofs}).

{\em Soundness} 
See Lem.~\ref{lemma:concluding} and 
 \figureword~\ref{fig:pair_elim_2_soundness_derivation} for the rule \rulename{\pairelimlabel{2}}.
\end{proof}

\begin{figure}
    \scalebox{0.90}{\begin{prooftree}
    \hypo{(A, X)}
    \infer1[\rulename{WO}]{(A, Y \vee X)}
    \infer1[\rulename{SI}]{(B \wedge A, Y \vee X)}
    \hypo{(B, Y \vee \neg X)}
    \infer1[\rulename{SI}]{(B \wedge A, Y \vee \neg X)}
    \infer2[\rulename{AND}]{(B \wedge A, (Y \vee X) \wedge (Y \vee \neg X))}
    \infer1[\rulename{WO}]{(B \wedge A, Y \vee (X \wedge \neg X))}
    \infer1[\rulename{WO}]{(B \wedge A, Y)}
    \hypo{\hspace{-18mm} (B \wedge \neg A, Y)}
    \infer2[\rulename{OR}]{((B \wedge A) \vee (B \wedge \neg A), Y)}
    \infer1[\rulename{SI}]{(B \wedge (A \vee \neg A), Y)}
    \infer1[\rulename{SI}]{(B, Y)}
    \end{prooftree}}

    \setlength{\belowcaptionskip}{-10pt}
    \caption{Derivation of the rule \rulename{\pairelimlabel{2}} in \ioslogic{2}}
    \label{fig:pair_elim_2_soundness_derivation}
\end{figure}
%


\subsection{Causal Production Inference \ioslogic{4}}
\label{sec:sequent_4}
The calculus \iocalc{4} for the causal version of reusable output logic \ioslogic{4} extends the core calculus (consisting of $LK$ with the the rules $(\inaxlabel)$ and $(\outaxlabel)$) with the pair elimination rule \rulename{\pairelimlabel{4}} in 
\figureword~\ref{fig:IO_pair_elimination_rules}.

Inspired by the normal modal logic embedding of \ioslogic{4} in~\cite{causal_logic},
the shape of the rule \rulename{\pairelimlabel{4}} requires to amend the statement of the characterization lemma. The proof of this lemma is similar to the one for \iocalc{2}
(See full proof in  Appendix~\ref{sec:ap_proofs}).

\begin{lemma}[Characterization lemma for \iocalc{4}]
\label{lem:char_4}
$\ioseq{(A_1, X_1), \dots, (A_n, X_n)}{B}{Y}$ is derivable in \iocalc{4} iff for all $(I, J)\in $ $\partset(\{1, \ldots ,n\})$, either $B, \{X_j\}_{j \in J} \seqlk
\{A_i\}_{i \in I}$ or ${\{X_j\}_{j \in J} \seqlk Y}$ is derivable in $LK$.
\end{lemma}

%

\begin{theorem}[Soundness and completeness of \iocalc{4}]
\label{th:sound_comp_4}
$\ioseq{G}{B}{Y}$ is derivable in \iocalc{4} iff the pair $(B, Y)$ is derivable from the pairs in $G$ in \ioslogic{4}.
\end{theorem}
\begin{proof}
{\em (Completeness)}
To derive \rulename{CT} in \iocalc{4} we first derive $\ioseq{(A, X), (A \wedge X, Y)}{A}{Y}$, and then apply \rulename{\cutlabel} and \rulename{\contractlabel} (\figureword~\ref{fig:IO_structural_rules}).
The claim follows by the admissibility of these
structural rules in \iocalc{4}, which can be reduced to the admissibility of the structural rules in LK as for \iocalc{2}.
(Full proof in Appendix~\ref{sec:ap_proofs}).

{\em (Soundness)}
Replace in \figureword~\ref{fig:pair_elim_2_soundness_derivation}
 the subtree that derives $(B \wedge A, Y \vee \neg X)$ from the pair $(B, Y \vee \neg X)$ by the following derivation, which uses the rule \rulename{CT}
\[ \scalebox{0.90}{\begin{prooftree}
    \hypo{(A, X)}
    \infer1[\rulename{SI}]{(B \wedge A, X)}
    \hypo{(B \wedge X, Y \vee \neg X)}
    \infer1[\rulename{SI}]{(B \wedge A \wedge X, Y \vee \neg X)}
    \infer2[\rulename{CT}]{(B \wedge A, Y \vee \neg X)}
\end{prooftree}} \]
\end{proof}

%
\subsection{Production Inference \ioslogic{1} and
Regular Production Inference
\ioslogic{3}}
\label{sec:sequent_13}
%

The calculi \iocalc{1}  and \iocalc{3} for the causal simple-minded output \ioslogic{1} and simple-minded reusable output \ioslogic{3} consist of $LK$ with $(\inaxlabel)$ and $(\outaxlabel)$ extended with the pair elimination rules \rulename{\pairelimlabel{1}} and \rulename{\pairelimlabel{3}} in 
\figureword~\ref{fig:IO_pair_elimination_rules}, respectively.

\begin{remark}
Unlike \ioslogic{2} and \ioslogic{4}, there is no modal embedding in the literature to provide guidance for the development of the pair elimination rules for \ioslogic{1} and \ioslogic{3}. 
These rules are instead designed by appropriately modifying \rulename{\pairelimlabel{2}} and \rulename{\pairelimlabel{4}}.
Indeed \ioslogic{1} and \ioslogic{3} impose restrictions on \ioslogic{2} and \ioslogic{4}, respectively, by prohibiting the combination of inputs. The rules \rulename{\pairelimlabel{1}} and \rulename{\pairelimlabel{3}} are defined by reflecting this limitation.
\end{remark}


Due to the $LK$ premise in their peculiar rules, derivations in \iocalc{1} and \iocalc{3} have a simpler form w.r.t. derivations in \iocalc{2} and \iocalc{4}; this form could be exploited for the soundness and completeness proof. We proceed instead as for the latter calculi by proving the characterization lemma. This lemma will be key to 
solve the open problems about computational bounds and modal embeddings 
for \ioslogic{1} and \ioslogic{3}. 
The proof of the lemma for \iocalc{1} and \iocalc{3} is less straightforward than for the other logics. The intuition here is that the characterization considers all possible ways to apply the rule \rulename{\pairelimlabel{3}} (or \rulename{\pairelimlabel{1}}), by partitioning the premises $(A_1, X_1), \dots, (A_n, X_n)$ into two disjoint sets ($I$ of remaining deriving pairs and $J$ of eliminated pairs).
We will focus on the lemma for \iocalc{3}, the one for \iocalc{1} being a simplified case (with a very similar proof). 
Its proof relies on the following result:

\begin{lemma}
\label{lem:semi_invert_3}
If $\ioseq{(A,X), G}{B}{Y}$ is derivable in \iocalc{3}, then so is $\ioseq{G}{B \wedge X}{Y \vee \neg X}$.
\end{lemma}
\begin{proof}
Easy induction on the length of the derivation. We proceed by case distinction on the last applied rule. 
(Full proof
in Appendix~\ref{sec:ap_proofs}).
\end{proof}

\begin{lemma}[Characterization lemma for \iocalc{3}]
\label{lem:char_3}
$\ioseq{(A_1, X_1), \dots, (A_n, X_n)}{B}{Y}$ is derivable in \iocalc{3} iff for all 
$(I, J) \in \partset(\{1, \dots , n\})$, 
one of the following holds:
\begin{itemize}
\item ${B, \{X_j\}_{j \in J} \seqlk A_i}$ is derivable in LK for some $i \in I$,
\item ${B, \{X_j\}_{j \in J} \seqlk}$ is derivable in LK,
\item ${\{X_j\}_{j \in J} \seqlk Y}$ is derivable in LK.
\end{itemize}
\end{lemma}
\begin{proof}

$(\Rightarrow)$: 
Let $(I, J) \in \partset(\{1, \dots , n\})$ be any partition. 
By (several application of) \lemmaword~\ref{lem:semi_invert_3}
to each $(A_j, X_j)$ with $j \in J$, we get that the I/O sequent
$(\ast) \; \ioseq{\{ (A_i, X_i) \mid i \in I \}}{B \wedge \bigwedge_{j \in J} X_j}{Y \vee \bigvee_{j \in J} \neg X_j}$ is derivable in \iocalc{3}. 
We consider the last rule $(r)$ applied in the derivation of $(\ast)$.   Three cases can arise:
\begin{itemize}
    \item $(r) =$ \rulename{\pairelimlabel{3}} eliminating the deriving pair $(A_k, X_k)$ with $k \in I$; hence $B \wedge \bigwedge_{j \in J} X_j \seqlk A_k$ (and therefore $B, \{ X_j \}_{j \in J} \seqlk A_k$) 
    is derivable in LK,
    \item $(r) =$ \rulename{\inaxlabel} then $B, \{ X_j \}_{j \in J} \seqlk$ is derivable in LK.
    \item $(r) =$ \rulename{\outaxlabel} then 
    ${\{ X_j \}_{j \in J} \seqlk Y}$ is derivable in LK.
    
\end{itemize}

$(\Leftarrow)$: We stepwise construct a derivation in \iocalc{3} of
$\ioseq{(A_1, X_1), \dots, (A_n, X_n)}{B}{Y}$. We start with
$(I,J) = (\{1,\dots,n\}, \emptyset)$, and we distinguish the three cases 
from the assumption of the lemma:
(1) ${B, \{X_j\}_{j \in J} \seqlk A_i}$ for some $i \in I$,
(2) ${B, \{X_j\}_{j \in J} \seqlk}$, and
(3) ${\{X_j\}_{j \in J} \seqlk Y}$. If either (2) or (3) holds the derivation follows by applying a concluding rule ($J = \emptyset$). If (1) holds, we apply \rulename{\pairelimlabel{3}} bottom up getting
$\ioseq{\{(A_t, X_t)\}_{t \in \{1,\dots,n\} \setminus \{i\}}}{B \wedge X_i}{Y \vee \neg X_i}$ as the second premise. We now apply the same reasoning to 
this latter sequent
(considering the partition $(I, J) = (\{1,\dots,n\} \setminus \{i\}, \{i\})$), and keep applying it 
for the second premise of \rulename{\pairelimlabel{3}}, until a concluding rule is applied. This will eventually happen since $I$ loses the index of the eliminated pair at each step.
\end{proof}


\vspace{-2mm}
The characterization lemma for \iocalc{1} has the following formulation (the proof is similar to the proof for \iocalc{3}).
%

\begin{lemma}[Characterization lemma for \iocalc{1}]
\label{lem:char_1}
$\ioseq{(A_1, X_1), \dots, (A_n, X_n)}{B}{Y}$ is derivable in \iocalc{1} iff for all partitions $(I, J) \in \partset(\{1, \dots , n\})$, at least one of the following holds:
\begin{itemize}
\item ${B \seqlk A_i}$ is derivable in LK for some $i \in I$,
\item ${B \seqlk}$ is derivable in LK,
\item ${\{X_j\}_{j \in J} \seqlk Y}$ is derivable in LK.
\end{itemize}
\end{lemma}



\begin{theorem}[Soundness and completeness of \iocalc{1} and \iocalc{3}]
\label{th:sound_comp_1_3}
$\ioseq{G}{B}{Y}$ is derivable in \iocalc{1} (\iocalc{3}) iff $(B, Y)$ is derivable from the pairs in $G$ in \ioslogic{1} (\iologic{3}).
\end{theorem}
\begin{proof}
({\em Completeness}) uses the admissibility proof of the structural rules in \figureword~\ref{fig:IO_structural_rules}
The case $\ioseq{(A,X)}{A}{X}$ is now handled as follows:
\[ \scalebox{0.90}{\begin{prooftree}
    \hypo{A \seqlk A}
    \hypo{\seqlk X \vee \neg X}
    \infer1[\rulename{\outaxlabel}]{\ioseq{}{A}{X \vee \neg X}}
    \infer2[\rulename{\pairelimlabel{1}}]{\ioseq{(A,X)}{A}{X}}
\end{prooftree}} \]
The derivation of the rule \rulename{SI} directly uses the premise $B \seqlk A$, instead of the rule \rulename{\inaxlabel}.
(Full proof
in \appendixword~\ref{sec:ap_proofs}).

{\em (Soundness)}: For \iocalc{1}, the derivation of  \rulename{\pairelimlabel{1}} is obtained by
 juxtaposing to the subderivation of $(B \wedge A, Y)$ 
 from $(A,X)$ and $(B,Y \vee \neg X)$ in Fig.~\ref{fig:pair_elim_2_soundness_derivation} the following  
 \vspace{-8pt}
 \[ \scalebox{0.90}{\begin{prooftree}
    \hypo{\vdots}
    \infer1{(B \wedge A, Y)}
    \infer1[\rulename{SI}]{(B, Y)}
\end{prooftree}}  \]
noticing that the side condition $B \seqlk A$ of the rule 
\rulename{\pairelimlabel{1}} implies that $B$ and $B \wedge A$ are classically equivalent.
The soundness of the rule \rulename{\pairelimlabel{3}} is very similar since the translation of the rule \rulename{\pairelimlabel{4}} from the proof of Theorem~\ref{th:sound_comp_4} also contains a sub-derivation of the pair $(B \wedge A, Y)$ from $(A,X)$ and $(B \wedge X, Y \vee \neg X)$ which uses only the rules of \iologic{3} (\rulename{WO}, \rulename{SI}, \rulename{AND} and \rulename{CT}).
\end{proof}

\section{Causal I/O Logics vs. Original I/O Logics}
\label{sec:io}
We establish a syntactic correspondence between derivability in the original I/O logics and in their causal version. This correspondence obtained utilizing the sequent calculi \iocalc{1}-\iocalc{4}, will enable to use them for \iologic{1}-\iologic{4}, and to transfer all results arising from the calculi for the causal I/O logics to the original I/O logics.

Note that  \iocalc{1}-\iocalc{4} rely on the axiom \rulename{BOT}, which is absent in the original I/O logics. 
An inspection of the soundness proofs for our calculi shows that \rulename{BOT} is solely employed in the translation of the rule \rulename{\inaxlabel} (Lemma~\ref{lemma:concluding}).
 Can we simply remove this rule and hence get rid of axiom \rulename{BOT}? Yes, but only for
 \iocalc{1} and \iocalc{3}, where, as evidenced by the completeness proof, \rulename{\inaxlabel} is used to derive \rulename{BOT} and not utilized elsewhere.
Hence, by removing the rule
\rulename{\inaxlabel} from \iocalc{1} and \iocalc{3} we get sequent calculi for \iologic{1} and
\iologic{3}. These calculi are close to
the sequent calculi inspired by conditional logics 
introduced in~\cite{IO_sequents}. The same does not hold for \iocalc{2} and \iocalc{4}, where \rulename{\inaxlabel} is needed, e.g., to derive \rulename{SI}.
Instead of developing ad hoc calculi to handle the original I/O logics, we leverage \iocalc{1}-\iocalc{4} 
using the following result: 

\begin{theorem}
\label{lem:original_causal_connection}
$(A_1, X_1), \dots, (A_n, X_n) \vdash_{\iologic{k}} (B, Y)$ 
iff $(A_1, X_1), \dots, (A_n, X_n) \vdash_{\ioslogic{k}} (B, Y)$ 
and ${X_1, \dots, X_n \seqcl Y}$ in classical logic, for each $k= 1, \dots , 4$.
\end{theorem}
\begin{proof}
$(\Rightarrow)$
Derivability in \iologic{k} implies derivability in the stronger logic \ioslogic{k}. The additional condition $X_1, \dots, X_n \seqcl Y$ can be proved by an easy induction on the length of the derivation in the original I/O logics: it is enough to check that for every rule if the outputs of all (pairs-)premises follow from ${X_1 \wedge \dots \wedge X_n}$, then so is the output of the (pair-)conclusion.

$(\Leftarrow)$
Consider our  sequent calculi \iocalc{1} - \iocalc{4}.
The translation constructed in their soundness theorem shows how to map a derivation of the I/O sequent $\ioseq{(A_1, X_1), \dots, (A_n, X_n)}{B}{Y}$ in \iocalc{k} 
into a I/O derivation of $(B,Y)$ from the pairs $\{(A_1, X_1), \dots, (A_n, X_n)\}$ using the rules of the logic $\ioslogic{k}$. As observed before, \rulename{BOT} appears in this transformed derivation only inside sub-derivations of the pairs $(B', Y')$ derived in \iocalc{k} by the rule $\rulename{\inaxlabel}$. We can replace every such sub-derivation with a derivation that does not contain the axiom \rulename{BOT} and uses only the rules of the weakest logic \iologic{1}. This latter derivation relies on the premise ${B' \seqlk}$ of the rule \rulename{\inaxlabel} (which implies $B' \seqcl \bot$), on the condition ${X_1 \land \dots \land X_n \seqcl Y}$ from the statement of the lemma and the fact that ${Y \seqcl Y'}$ since in all our calculi  the goal output in the premises of the elimination rules is the same or weaker than the goal output in the conclusion. The required derivation 
is 
the following:
\[
\centering
    \scalebox{0.90}{\begin{prooftree}
    \hypo{(A_1, X_1)}
    \infer1[\rulename{SI}]{(\bot, X_1)}
    \hypo{\dots}
    \hypo{(A_n, X_n)}
    \infer1[\rulename{SI}]{(\bot, X_n)}
    \infer3[\rulename{AND} $\times(n-1)$]{(\bot, X_1 \wedge \dots \wedge X_n)}
    \infer1[\rulename{WO}]{(\bot, Y')}
    \infer1[\rulename{SI}]{(B', Y')}
    \end{prooftree}}
\]
After the replacement of all indicated sub-derivation of $(B', Y')$ with the ones obove, we will get a derivation of $(B, Y)$ that does not use the axiom \rulename{BOT} and thus $(B, Y)$ is derivable in the original I/O calculus in~\cite{IO_original} for \iologic{k}.
\end{proof}
\begin{remark}
The constructive proof above heavily relies
on the restricted form of the I/O derivations resulting from translating our sequent derivations.
If at all possible, finding ways to eliminate the use of the \rulename{BOT} axiom in arbitrary I/O derivations within \ioslogic{k} would be a challenging task. 
The power of structural proof theory lies in its capacity to solely examine 
well-behaved derivations.
\end{remark}

\section{Applications}
Our proof-theoretic investigation is used here
to establish the following results for the original and the causal I/O logics:
uniform possible worlds semantics (Sec.~\ref{sec:semantics}), co-NP-completeness and automated deduction methods (Sec.~\ref{sec:AD}), and new embeddings into normal modal logics (Sec.~\ref{sec:modal_embeddings}).

\subsection{Possible Worlds Semantics}
\label{sec:semantics}
We provide possible worlds semantics for both the original and the causal I/O logics. Our semantics is a generalization of the bimodels semantics in~\cite{causality_for_nonmonotonic} for \ioslogic{2}; it turns out to be
simpler than them for the remaining causal logics, and than
the procedural semantics for the original I/O logics.
As we will see, this semantics facilitates clean and uniform solutions to various unresolved inquiries regarding I/O logics that were only partially addressable.

First, notice that a contrapositive reading of the characterization lemmas leads to countermodels for non-derivable statements in all considered causal I/O logics. These countermodels consist of (a partition and) several boolean interpretations (two for \ioslogic{2}, \ioslogic{4} and their causal versions, and $(n+2)$ for \ioslogic{1}, \ioslogic{3} and their causal versions) that falsify the $LK$ sequents from the respective lemma statement. 
We show below that a suitable generalization of these countermodels provides alternative semantic characterizations for both the original and the causal I/O logics.


A possible worlds semantics for the causal I/O logics was introduced 
by~\cite{causality_for_nonmonotonic} using \emph{bimodels}. For the simplest case of \ioslogic{2}, bimodel is a pair of worlds (here `world' can be seen as a synonym for boolean interpretation) corresponding to input and output states. 
\begin{definition}\cite{causality_for_nonmonotonic}
A pair $(A,X)$ is said to be \emph{valid} in a bimodel $(\inputworld, \outputworld)$ if $\inputworld \vDash A$ implies $\outputworld \vDash X$.
\end{definition}
The adequacy of this semantics implies, in particular, that $G \vdash_{\ioslogic{2}}(B, Y)$ 
if and only if the validity of all pairs from $G$ implies validity of $(B,Y)$ for all bimodels. 
The notion of bimodels for \ioslogic{1}, \ioslogic{3} and \ioslogic{4} is more complex, 
with
input and output states consisting of arbitrary deductively closed sets of formulae, instead of worlds.

To construct our semantics, 
we look at the countermodels provided by the characterization lemma from the point of view of the simplest bimodels of \ioslogic{2}. Lemma~\ref{lem:char_2} says indeed that if $\ioseq{(A_1, X_1), \dots, (A_n, X_n)}{B}{Y}$ is not derivable in \iocalc{2} there is a partition $(I, J)$, and two boolean interpretations $\inputworld$ and $\outputworld$ such that: $\inputworld$  falsifies the $LK$-sequent ${B \seqlk \{A_i\}_{i \in I}}$ (meaning that $\inputworld \vDash B$ and $\inputworld \nvDash A_i$ for all $i \in I$) and $\outputworld$  falsifies the $LK$-sequent ${\{X_j\}_{j \in J} \seqlk Y}$ (meaning that $\outputworld \vDash X_j$ for all $j \in J$ and $
\outputworld \nvDash Y$). These two interpretations lead to a bimodel that falsifies $\ioseq{(A_1, X_1), \dots, (A_n, X_n)}{B}{Y}$; indeed all pairs $(A_i, X_i)$ for $i \in I$ are valid in $(\inputworld, \outputworld)$ as $\inputworld \nvDash A_i$, all pairs $(A_j, X_j)$ for $j \in J$ are valid in $(\inputworld, \outputworld)$ as $\outputworld \vDash X_j$, but $(B, Y)$ is not valid in $(\inputworld, \outputworld)$ because $\inputworld \vDash B$ and $\outputworld \nvDash Y$.

Reasoning in a similar way about the
countermodels for \ioslogic{1} given by \lemmaword~\ref{lem:char_1}, 
we observe there are now multiple input worlds, each falsifying in $B \seqlk A_i$ 
the input $A_i$ (plus one additional input world that arises from the sequent $B \seqlk$). This leads to the following generalization of bimodels with multiple input worlds.

\begin{definition}
An \emph{I/O model} is a pair $(\inputworldset, \outputworld)$ where $\outputworld$ is the output world, and $\inputworldset$ is a set of input worlds.
\end{definition}

The definition of validity in an I/O model will be modified to require that the input formula is true in all input worlds, rather than just in the unique input world. This update ensures that the existence of a single input world falsifying $A$ is enough to establish the validity of the pair $(A, X)$. Moreover, the additional ability to reuse outputs as inputs in the logics \ioslogic{3} and \ioslogic{4} can be expressed in these models by the requirement that a triggered output $X$ should hold in the input worlds too. This leads to the following two definitions of validity in I/O models -- one for the logics \ioslogic{1} and \ioslogic{2}, the other for \ioslogic{3} and \ioslogic{4}.

\begin{definition}
\label{def:model_validity} $ $ 
\begin{itemize}
    \item An I/O pair $(A,X)$ is \emph{\otvalid} in an I/O model $(\inputworldset, \outputworld)$ if $(\forall \inputworld \in \inputworldset.\; \inputworld \vDash A)$ implies $\outputworld \vDash X$.
    \item An I/O pair $(A,X)$ is \emph{\tfvalid} in an I/O model $(\inputworldset, \outputworld)$ if $(\forall \inputworld \in \inputworldset.\; \inputworld \vDash A)$ implies ${(\forall w \in \{\outputworld\} \cup \inputworldset.\; w \vDash X)}$.
\end{itemize}
\end{definition}

When clear from the context, henceforth we will use the term {\em validity} to mean either \otvalidity (hence referring to the logics \iologic{1}, \iologic{2}, \ioslogic{1}, and \ioslogic{2}) or \tfvalidity (hence referring to \iologic{3}, \iologic{4}, \ioslogic{3}, and \ioslogic{4}). 


\begin{table}
\centering

\begin{tabular}{|c|c|c|}
\hline
Logic & Frame condition & Notion of validity \\
\hline
\ioslogic{1} & $|\inputworldset| \ge 1$ & \otvalidity \\ 
\ioslogic{2} & $|\inputworldset|$ = 1 & \otvalidity \\  
\ioslogic{3} & $|\inputworldset| \ge 1 $ & \tfvalidity \\ 
\ioslogic{4} & $|\inputworldset|$ = 1 & \tfvalidity \\  
\hline
\iologic{1} & no conditions & \otvalidity \\ 
\iologic{2} & $|\inputworldset| \le 1$ & \otvalidity \\  
\iologic{3} & no conditions & \tfvalidity \\ 
\iologic{4} & $|\inputworldset| \le 1$ & \tfvalidity \\  
\hline
\end{tabular}
\setlength{\belowcaptionskip}{-10pt}
\caption{Conditions on I/O models (size of the set $\inputworldset$ of input worlds) and corresponding notions of validity for I/O models.} 
\label{tab:IO_models_conditions}
\end{table}

\begin{proposition}[Semantics for \ioslogic{k}]
\label{prop:IO_models_sound_complete_causal}
$G \vdash_{\ioslogic{k}} (B, Y)$ iff for all I/O models (satisfying the conditions in~\tableword~\ref{tab:IO_models_conditions}) validity of all pairs from $G$ implies validity of $(B,Y)$. 
\end{proposition}
\begin{proof}
The proofs rely on the characterization lemmas for each of the four logics. We detail the proof for  \ioslogic{3}, the others being similar. The equivalence is proved by negating both statements.

$(\Leftarrow)$ Suppose $\ioseq{(A_1, X_1), \dots, (A_n, X_n)}{B}{Y}$ is not derivable in \iocalc{3} (and hence in \ioslogic{3}). We show the existence of a countermodel. By \lemmaword~\ref{lem:char_3} there exists a partition $(I, J)$ such that the sequents 
${\{X_j\}_{j \in J} \seqlk Y}$, ${B, \{X_j\}_{j \in J} \seqlk}$ and ${B, \{X_j\}_{j \in J} \seqlk A_i}$ for all $i \in I$ are not derivable in LK. By the soundness and completeness of $LK$ w.r.t. the classical truth tables semantics there exist:
\begin{itemize}
\item an interpretation $\outputworld$, s.t.~$\outputworld \vDash X_j$ $\; \forall j \in J$ and $\outputworld \nvDash Y$,
\item an interpretation $\inputworld_0$, s.t.~$\inputworld_0 \vDash B$ and $\inputworld_0 \vDash X_j$ $\; \forall j \in J$,
\item for every $i \in I$ an interpretation $\inputworld_i$, s.t. $\inputworld_i \vDash B$, $\inputworld_i \vDash X_j$ for all $j \in J$ and $\inputworld_i \nvDash A_i$.
\end{itemize}
Then the I/O model $M = {(\{\inputworld_0\} \cup \{\inputworld_i \mid i \in I\}, \outputworld)}$ (which has at least one input world) is a countermodel for the derivability of the I/O sequent $\ioseq{(A_1, X_1), \dots, (A_n, X_n)}{B}{Y}$. Indeed, every pair $(A_i, X_i)$ for $i \in I$ is 3-4-valid\ in $M$ since $A_i$ is not true in the input world $\inputworld_i$; every pair $(A_j, X_j)$ for $j \in J$ is 3-4-valid\ in $M$ since $X_j$ is true in all worlds of $M$; but the pair $(B, Y)$ is not 3-4-valid\, since $B$ is true in all input worlds, while $Y$ is not true in the output world $\outputworld$.

$(\Rightarrow)$ Let $M = (\inputworldset, \outputworld)$ (with $\inputworldset \neq \emptyset$) be a countermodel. I.e. $M$ 3-4-validates all pairs $\{(A_1, X_1),\dots,(A_n, X_n)\}$, but does not 3-4-validate the pair $(B, Y)$. We prove that the sequent $\ioseq{(A_1, X_1) \dots, (A_n, X_n)}{B}{Y}$ is not derivable in \ioslogic{3}. 
We show the existence of a partition $(I, J)$, such that none of the $LK$ sequents ${\{X_j\}_{j \in J} \seqlk Y}$, ${\{X_j\}_{j \in J} \seqlk Y}$, and ${B, \{X_j\}_{j \in J} \seqlk}$ is derivable. The claim follows then by \lemmaword~\ref{lem:char_3}.
For such partition, we chose $J$ as $\{ j \colon \forall \inputworld \in \inputworldset.\; \inputworld \vDash A_j \}$, and $I$ as the rest of the indices. Notice that $X_j$ is true in all worlds of $M$ for every $j \in J$ by definition of 3-4-validity of $(A_j, X_j)$. Also, the fact that $(B, Y)$ is not \tfvalid\ in $M$ means that $B$ holds in all input worlds, but there exists a world $w^{\ast} \in  \{\outputworld\} \cup \inputworldset$, s.t. $w^{\ast} \nvDash Y$. Then:
\vspace{-2mm}
\begin{itemize}
\item ${\{X_j\}_{j \in J} \seqlk Y}$ is not derivable in LK, because this sequent does not hold in the world $w^{\ast}$.
\item ${B, \{X_j\}_{j \in J} \seqlk}$ is not derivable in LK, because this sequent does not hold in any input world of $M$ (and there is at least one by the condition $\inputworldset \neq \emptyset$).
\item For any $i \in I$ there exists an input world $\inputworld_i$, s.t. ${\inputworld_i \nvDash A_i}$ (by the choice of $I$). Hence ${B, \{X_j\}_{j \in J} \seqlk A_i}$ is not derivable in $LK$, as this sequent does not hold in $\inputworld_i$. \qedhere
\end{itemize}
\end{proof}

Dropping the condition of having at least one input world leads to models for the original I/O logics.

\begin{proposition}[Semantics for \iologic{k}]
\label{prop:IO_models_sound_complete_original}
$G \vdash_{\iologic{k}} (B, Y)$
iff for all I/O models (satisfying the conditions in~\tableword~\ref{tab:IO_models_conditions})  validity of all pairs from $G$ implies validity of $(B,Y)$.
\end{proposition}
\begin{proof}
By 
\theoremword~\ref{lem:original_causal_connection}, $(B,Y)$ is derivable from  $(A_1, X_1), \dots,$ $  (A_n, X_n) $ in \iologic{k} iff it is derivable in \ioslogic{k} together with the additional condition  $X_1, \dots, X_n \seqcl Y$. We prove that this additional condition is equivalent to the fact that every model with zero input worlds that validates all pairs from $G$ also validates $(B,Y)$. Notice that this will prove the proposition, as the only difference between the proposed semantics for a causal I/O logic and the corresponding original one is that the latter additionally considers models with zero input worlds (see \tableword~\ref{tab:IO_models_conditions}).


For both notions of validity, the validity of a pair $(A,X)$ in $(\emptyset, \outputworld)$ is equivalent to $\outputworld \vDash X$.
Now, $X_1, \dots, X_n \seqcl Y$ means that every interpretation that satisfies every $X_i$ also satisfies $Y$, which is equivalent to the fact that every model $(\emptyset, \outputworld)$ (with arbitrary interpretation $\outputworld$) that validates every $(A_i, X_i)$ also validates $(B,Y)$. \end{proof}

\begin{remark}
A natural interpretation for the I/O models in the deontic context regards input world(s) as (different possible instances of) the real world, and the output world as the ideal world, where all triggered obligations are fulfilled. 
\end{remark}


\subsection{Complexity and Automated Deduction}
\label{sec:AD}
We investigate the computational properties of the four original I/O logics and their causal versions. One corollary of our previous results is
co-NP-completeness for all of them. Moreover, we can explicitly reduce the entailment problem in all
these logics to the (un-)satisfiability of one classical propositional formula of polynomial size, a thoroughly studied problem with a huge variety of efficient tools available.


\begin{corollary}
The entailment problem is a co-NP-complete problem for all eight considered I/O logics.
\end{corollary}
\begin{proof}
The characterization lemmas for the logics \ioslogic{1}-\ioslogic{4} imply that the non-derivability of a pair from $n$ pairs can be non-deterministically verified in polynomial time by guessing the non-fulfilling partition (consisting of $n$ bits) and then non-deterministically checking the non-derivability of all sequents (at most $(n + 2)$) for this partition; the latter task can be done in linear time.
For the original I/O logics, by \theoremword~\ref{lem:original_causal_connection} we also need to verify 
that the additional condition does not hold (guessing a falsifying boolean assignment). Thus, the entailment problem belongs to co-NP for all considered I/O logics.
The co-NP-completeness follows by the fact that any arbitrary propositional formula $Y$ is classically valid iff $(\top,Y)$ can be derived from no pairs in any calculus for the considered logics (notice that the additional condition of \theoremword~\ref{lem:original_causal_connection} also boils down to the classical validity of $Y$).
\end{proof}


We provide an explicit reduction the derivability in I/O logics to the classical validity. 
For \iologic{2} and \iologic{4} this is already contained in~\cite{IO_original}. 
Using the semantics introduced in Sec.~\ref{sec:semantics}, we obtain this result for Bochmann's causal I/O logics and their original version in a uniform way.

\propword~\ref{prop:IO_models_sound_complete_causal} shows that the underivability of $\ioseq{G}{B}{Y}$ in the causal I/O logics is equivalent to the existence of an I/O model that validates all pairs in $G$, but does not validate $(B,Y)$. For \ioslogic{2} and \ioslogic{4} a countermodel should have exactly one input world, while for \ioslogic{1} and \ioslogic{3} 
there is always one with at most $(|G| + 1)$ input worlds. 

We will encode existence of a countermodel to $\ioseq{G}{B}{Y}$
with exactly $\numofinputworlds{k}$ input worlds (with $\numofinputworlds{k} = 1$ for $k=2,4$, and $\numofinputworlds{k} = |G| + 1$ for $k=1,3$) in classical logic. For the encoding, we assign to the input worlds the numbers from $1$ to $\numofinputworlds{k}$, and $0$ to the output world. Let $\varset$ be the finite set of all propositional variables that occur in the formulae of $G$ or $(B,Y)$. For every variable $x \in \varset$, our encoding will use $(\numofinputworlds{k} + 1)$ copies of this variable $\{ x^0, \dots, x^{\numofinputworlds{k}} \}$ with the intuitive interpretation that $x^l$ is true iff $x$ is true in the world number $l$. For an arbitrary formula $A$ with variables from $\varset$, let us denote by $A^l$ the copy of $A$ in which every variable $x \in \varset$ is replaced by its labeled version $x^l$. 
We read the formula $A^l$ as ``$A$ is true in the world number $l$''. The exact connection with $\ioseq{G}{B}{Y}$ is stated below.

\begin{lemma}
\label{lem:prop_reduction_causal}
${(A_1, X_1), \dots , (A_n, X_n) \vdash_{\ioslogic{k}} (B, Y)}$
iff the classical propositional formula ${\neg \propencoding{k}{n}{(B,Y)} \wedge \bigwedge\limits_{(A,X) \in G} \propencoding{k}{n}{(A,X)}}$ is unsatisfiable, where
\begin{itemize}
\item $\propencoding{k}{n}{(A,X)} = (\bigwedge\limits_{l = 1}^{\numofinputworlds{k}} A^l) \to X^0$\;\; for $k = 1,2$ 
\item $\propencoding{k}{n}{(A,X)} = (\bigwedge\limits_{l = 1}^{\numofinputworlds{k}} A^l) \to (\bigwedge\limits_{l = 0}^{\numofinputworlds{k}} X^l)$\;\; for $k = 3,4$ 
\end{itemize}
\end{lemma}
\begin{proof}
We prove the contrapositive version.

($\Rightarrow$) Let $\varset^L$ the set of all labeled copies of variables in $\varset$ ($\varset^L = {\{ x^l \,\mid\, x \in \varset, l \in \{0,\dots,\numofinputworlds{k}\}\}}$). Suppose there is a valuation ${v \colon \varset^L \to \{0,1\}}$, that satisfies the formula in the statement (i.e., $v \vDash \propencoding{k}{n}{(A,X)}$ for all $(A,X) \in G$ and $v \nvDash \propencoding{k}{n}{(B,Y)}$). $v$ can be decomposed into ${(\numofinputworlds{k} + 1)}$ valuations $v_l \colon \varset \to \{0,1\}$, one for each label ($v_l(x) = v(x^l)$).
It is easy to see that $(\ast)$: For every formula $A$ with variables in $\varset$, $v \vDash A^l$ iff $v_l \vDash A$ (it can be proven by trivial induction).
The valuations $\{v_l\}$ can then be turned into an I/O model $M = (\{v_1, \dots, v_{\numofinputworlds{k}}\}, v_0)$. Then using the reading of $v \vDash A^l$ given by $(\ast)$ we can see that $v \vDash \propencoding{k}{n}{(A,X)}$ (for $k=1,2$) iff $(A,X)$ is 1-2-valid in $M$, and $v \vDash \propencoding{k}{n}{(A,X)}$ (for $k=3,4$) iff $(A,X)$ is 3-4-valid in $M$.
Therefore, since $v$ satisfies the formula in the statement, $M$ validates all pairs from $G$ and does not validate $(B,Y)$, which implies that $\ioseq{G}{B}{Y}$ is not derivable in \ioslogic{k}.

($\Leftarrow$) Here instead of decomposing a valuation of labeled variables into $(\numofinputworlds{k} + 1)$ worlds, we use a countermodel $(\{\inputworld_1, \dots, \inputworld_{\numofinputworlds{k}}\}, \outputworld)$ to define a valuation ${v \colon \varset^L \to \{0,1\}}$ of labeled variables (with $v(x^0) = out(x)$ and $v(x^l) = in_l(x)$). The proof proceeds as in the other direction. 
\end{proof}

The result is extended to the original I/O logics via
Th.~\ref{lem:original_causal_connection}.

\begin{lemma}
\label{lem:prop_reduction_original}
${(A_1, X_1), \dots , (A_n, X_n) \vdash_{\iologic{k}} (B, Y)}$ iff the classical propositional formula 
$\mathcal{F}^k_n \vee {(\neg Y \wedge \bigwedge\limits_{i=1}^{n} X_i)}$ is  unsatisfiable, where $\mathcal{F}^k_n$ is the formula encoding derivability of $\ioseq{(A_1, X_1), \dots, (A_n, X_n)}{B}{Y}$ in \ioslogic{k} from \lemmaword~\ref{lem:prop_reduction_causal}.
\end{lemma}
\begin{proof}
The disjunct $(\neg Y \wedge (X_1 \wedge \dots \wedge X_n))$ arises
from \theoremword~\ref{lem:original_causal_connection} (derivability in \iologic{k} is equivalent to derivability in \ioslogic{k} {\em and} the classical entailment of $Y$ from $\{X_i\}_{i=1}^n$). The claim follows by
\lemmaword~\ref{lem:prop_reduction_causal}.
\end{proof}


%

\subsection{Embeddings into Normal Modal Logics}
\label{sec:modal_embeddings}

We provide uniform embeddings into normal modal logics.
The embeddings are a corollary of soundness and completeness of I/O logics w.r.t. I/O models.

More precisely we show that $\ioseq{G}{B}{Y}$ in I/O logics iff a certain sequent consisting of shallow formulae only (meaning that the formulae do not contain nested modalities) is valid in suitable normal modal logics. To do that we establish a correspondence between pairs and shallow formulae.

\begin{figure}
    \centering
    \scalebox{0.65}{
    \begin{tikzpicture}[modal]
	\node[world] (out) {$\outputworld$};
	\node[world] (in1) [above =3mm of out, xshift=-18mm] {$\inputworld_1$};
	\node[world] (inN) [above =3mm of out, xshift=18mm] {$\inputworld_N$};;
        \node (dots) [above=7mm of out] {$\dots$};
	\path[->] (out) edge (in1);
	\path[->] (out) edge (inN);
	\path[->] (in1) edge[reflexive left=7] (in1);
	\path[->] (inN) edge[reflexive right=7] (inN);

    \end{tikzpicture} }
    \setlength{\belowcaptionskip}{-10pt}
    \caption{Kripke countermodel construction for I/O model $(\{\inputworld_1,\dots,\inputworld_N\}, \outputworld)$. Arrows represent the accessibility relation.}
    \label{fig:Kripke_countermodel_construction}
\end{figure}
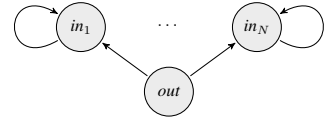
The I/O models already use the terminology of Kripke semantics that define normal modal logic. 
To establish a precise link between the two semantics we need only to define the accessibility relation on worlds. We will treat the set of input worlds $\inputworldset$ as the set of worlds accessible from the output world $\outputworld$ (see \figureword~\ref{fig:Kripke_countermodel_construction}). Under this view on input worlds, \otvalidity\ (resp. \tfvalidity\ ) of the pair $(A,X)$ is equivalent to the truth of the modal formula ${\square A \to X}$ (resp. ${\square A \to X \wedge \square X}$) in the world $\outputworld$.

Also, the conditions on the number of input worlds that are used in Prop.~\ref{prop:IO_models_sound_complete_causal} and Prop~\ref{prop:IO_models_sound_complete_original} to distinguish different I/O logics can be expressed in normal modal logics by standard Hilbert axioms. Specifically, axiom $\mathbf{D} \colon \square A \to \Diamond A$ forces Kripke models to have at least one accessible world, while ${\mathbf F} \colon \Diamond A \to \square A$ forces them to have at most one accessible world.
As proved below, the embedding works for
the basic modal logic $\mathbf{K}$ extended with $\mathbf{D}$ (which results in the well-known standard deontic logic~\cite{SDL} $\mathbf{KD}$), with $\mathbf{F}$, or both axioms.

Henceforth we abbreviate, e.g., validity in the logics $\mathbf{K}$ (respectively $\mathbf{K+F}$) with 
$\models_{\mathbf{K}/\mathbf{K+F}}$.
\begin{theorem}
$(B,Y)$ is derivable from pairs $G$ in

\begin{itemize}
   \item \iologic{1} and \iologic{2} iff $\; G^\square_{1/2} \models_{\mathbf{K}/\mathbf{K+F}} \square B \to Y$ 
    \item \iologic{3} and \iologic{4} iff  $\; G_{3/4}^\square \models_{\mathbf{K}/\mathbf{K+F}} \square B \to Y \! \wedge \! \square Y$ 
    
    \item \ioslogic{1} and \ioslogic{2} iff $\; G^\square_{1/2} 
    \models_{\mathbf{KD}/\mathbf{KD+F}} \square B \to Y$ 

    \item \ioslogic{3} and \ioslogic{4} iff $\; G_{3/4}^\square \models_{\mathbf{KD}/\mathbf{KD+F}}  \square B \to Y \! \wedge \! \square Y$ 
\end{itemize}
where $G_{1/2}^\square = \{ \square A_i \to X_i \mid (A_i,X_i) \in G \}$, \\ and 
$G_{3/4}^\square = \{ \square A_i \to X_i \wedge \square X_i \mid (A_i,X_i) \in G \}$.
\end{theorem}
\begin{proof}

We show these equivalences by transforming the I/O countermodels given by \propositionword~\ref{prop:IO_models_sound_complete_causal} and \propositionword~\ref{prop:IO_models_sound_complete_original} into Kripke countermodels for the corresponding modal logic and vice versa. The transformations will be the same for all the considered logics. We detail the case of \ioslogic{3}. 

($\Leftarrow$) Assume ${G_{3}^\square \models_{\mathbf{KD}} \square B \to Y \wedge \square Y}$ does not hold. Then there should exist a Kripke model $M$ in which every world has at least one world accessible from it, and a world $w$ in $M$, such that $w \vDash {\square A \to X \wedge \square X}$ for every ${(A,X) \in G}$ and $w \nvDash {\square B \to Y \wedge \square Y}$. Let $N(w)$ be the set of all worlds accessible from $w$ in $M$. Then the I/O model $(N(w), w)$ will be a countermodel for $\ioseq{G}{B}{Y}$; notice indeed that ${w \vDash \square A \to X \wedge \square X}$ means exactly \tfvalidity of $(A,X)$ in $(N(w), w)$, so all pairs in $G$ are \tfvalid in $(N(w), w)$, but $(B,Y)$ is not \tfvalid, and $|N(w)| \ge 1$. Hence $\ioseq{G}{B}{Y}$ is not derivable in \ioslogic{3}.

($\Rightarrow$) Assume $\ioseq{G}{B}{Y}$ is not derivable in \ioslogic{3}. Then there is some I/O model $(\inputworldset, \outputworld)$ with $|\inputworldset| \ge 1$, s.t. all pairs from $G$ are \tfvalid in $(\inputworldset, \outputworld)$ and $(B,Y)$ is not \tfvalid in $(\inputworldset, \outputworld)$. Consider the Kripke model $M$ that consists of worlds ${\inputworldset \cup \{\outputworld\}}$ with accessibility relation defined as shown in \figureword~\ref{fig:Kripke_countermodel_construction} (all input worlds are accessible from the output world and every input world is accessible from itself). $M$ satisfies the frame condition for $\mathbf{KD}$
as there is at least one accessible world from $out$ because of $|\inputworldset| \ge 1$,
and exactly one accessible world for every input world (itself). $\outputworld$ satisfies the modal formula ${A \to X \wedge \square X}$ (with $A$ and $X$ being propositional formulae) in the Kripke model $M$ iff the I/O pair $(A,X)$ is ${\text{\tfvalid}}$ in $(\inputworld, \outputworld)$. So in $M$, ${\outputworld \vDash A \to X \wedge \square X}$ for every pair $(A,X) \in G$ and ${\outputworld \nvDash B \to Y \wedge \square Y}$. Therefore ${G_{3}^\square \models_{\mathbf{KD}} \square B \to Y \wedge \square Y}$ does not hold.
\end{proof}

\begin{remark}
\label{rem:known_modal_embeddings}
Modal embeddings were already known for the causal logics \ioslogic{2} and \ioslogic{4}. 
The embedding for \ioslogic{2} was translating the pair $(A,X)$ into the {\bf K} 
formula $A \to \square X$. In \cite{IO_original} this embedding was stated
for \iologic{2} and \iologic{4} together with the additional condition $X_1, \dots , X_n \seqcl Y$
(appearing in our Th.~\ref{lem:original_causal_connection}).
Note that moving the modality to inputs allows for a more refined embedding. 
The validity of the statement ${(A \to \square X), \dots, (A \to \square X) \models (B \to \square Y)}$ 
is indeed the same in all four target logics we use ({\bf K}, {\bf KD}, {\bf K} + {\bf F} and {\bf KD} + {\bf F}), while the validity of $(\square A \to X),$ $ \dots, (\square A \to X) \models (\square B \to Y)$ is different.


\end{remark}

\section{Conclusions}

We have introduced 
sequent calculi for I/O logics. Our calculi provide a natural syntactic connection between derivability in the four original I/O logic~\cite{IO_original} and in their causal version~\cite{causality_for_nonmonotonic}.
Moreover, the calculi yield natural possible worlds semantics,
complexity bounds, embeddings into normal modal logics, as well as efficient deduction methods. 
It is worth noticing that
 our methods for the entailment problem offer derivability certificates (i.e., derivations) or counter-models as solutions. The efficient discovery of the latter can be accomplished using SAT solvers, along the line of~\cite{SAT_from_sequents}.
The newly introduced possible worlds semantics might be used to import in I/O logics tools and results from standard modal theory.

Our work encompasses many scattered results and presents uniform solutions to various unresolved problems; among them, it contains first proof-search oriented calculi for \ioslogic{2} and \ioslogic{4}; it provides a missing\footnote{
From ~\cite{IOhandbook}: "\textit{As a matter of facts, there is no direct (formal) connection between the semantics Bochman proposes and the operational semantics for I/O logic. 
The linkage between the two is established through the axiomatic characterization: both the possible-worlds semantics and the operational semantics give rise to almost the same axiom system}"} direct formal connection between the semantics of the original and the causal I/O logics; 
it introduces a uniform embedding into normal modal logics, that also applies to \iologic{1} and \iologic{3}, despite the absence in these logics of the \rulename{OR} rule\footnote{
\textit{From~\cite{IO_original}: ``As far as the authors are aware, it is not possible to characterise the system of simple-minded output (with or without reusability) by relabeling or modal logic in a straightforward way. The \rulename{OR} rule appears to be needed, so that we can work with complete sets}.''}; moreover, it
settles the complexity of the logics \iologic{3} and \ioslogic{3}. 
The latter logic has been used in~\cite{actual_causality} as the base for actual causality and in~\cite{causality_for_nonmonotonic}, together with \ioslogic{4}, to characterize strong equivalence of causal theories w.r.t. two different semantics: general and causal non-monotonic semantics. Strong equivalence is an important notion as theories satisfying it are ‘equivalent forever’, that is, they are interchangeable in any larger causal theory without changing the general/causal non-monotonic semantics. 
Furthermore \iologic{4} has been used as a base for formalizing legal concepts~\cite{kelsenian_logic}.
The automated deduction tools we have provided might be used also in these contexts. 

In this paper, we have focused on {\em monotonic} I/O logics. However, due to their limitations in addressing different aspects of causal reasoning~\cite{causality_book} and of normative reasoning, several non-monotonic extensions have been introduced. For example~\cite{IO_original_constrained,IO_minimal_aggregative} have proposed non-monotonic extensions that have also been applied to represent and reason about legal knowledge bases, as demonstrated in the work by Robaldo et al.~\cite{RobaldoBPRML20}.
%
Our new perspective on the monotonic I/O logics contributes to increase their understanding and can provide a solid foundation for exploring non-monotonic extensions.

\section*{Acknowledgements}
%
Work partially supported by the European Union's Horizon 2020 research and innovation programme under grant agreement No 101034440, and the WWTF project ICT22-023. 

%

\appendix

\bibliographystyle{kr}
\bibliography{biblio}

\newpage
\onecolumn

\section{LK sequent calculus}
\label{sec:ap_LK}

The sequent calculus LK is a deductive system for (first-order) classical logic, introduced by~\citeap{Gentzen}. Its basic objects are sequents, that are derivability assertions of the form $${A_1, \dots, A_n \seqlk B_1, \dots, B_m},$$ where ${\{A_1, \dots, A_n\}}$ and ${\{B_1, \dots, B_m\}}$ are multisets of formulae (assumptions, and  conclusions, respectively).

Axioms and rules of (the propositional fragment of) LK are as follows ($\Gamma$ and $\Delta$ denote arbitrary multisets of formulae):



\vspace{0.5cm}

\begin{center}

    \begin{tabular}{ccc}
\begin{prooftree}
    \hypo{}
    \infer1[\rulename{AX}]{\Gamma, x \seqlk x, \Delta}
\end{prooftree} & \hspace{10mm}
\begin{prooftree}
    \hypo{}
    \infer1[\rulename{$\bot$ L}]{\Gamma, \bot \seqlk \Delta}
\end{prooftree} & \hspace{10mm}
\begin{prooftree}
    \hypo{}
    \infer1[\rulename{$\top$ R}]{\Gamma \seqlk \top, \Delta}
\end{prooftree} \\[4mm]
    \end{tabular}
    \begin{tabular}{cc}
\begin{prooftree}
    \hypo{\Gamma, A, B \seqlk \Delta}
    \infer1[\rulename{$\wedge$ L}]{\Gamma, A \wedge B \seqlk \Delta}
\end{prooftree} &
\begin{prooftree}
    \hypo{\Gamma \seqlk A, \Delta}
    \hypo{\Gamma \seqlk B, \Delta}
    \infer2[\rulename{$\wedge$ R}]{\Gamma \seqlk A \wedge B, \Delta}
\end{prooftree} \\[4mm]
\begin{prooftree}
    \hypo{\Gamma, A \seqlk \Delta}
    \hypo{\Gamma, B \seqlk \Delta}
    \infer2[\rulename{$\vee$ L}]{\Gamma, A \vee B \seqlk \Delta}
\end{prooftree} &
\begin{prooftree}
    \hypo{\Gamma \seqlk A, B, \Delta}
    \infer1[\rulename{$\vee$ R}]{\Gamma \seqlk A \vee B, \Delta}
\end{prooftree} \\[4mm]
\begin{prooftree}
    \hypo{\Gamma \seqlk A, \Delta}
    \hypo{\Gamma, B \seqlk \Delta}
    \infer2[\rulename{$\to$ L}]{\Gamma, A \to B \seqlk \Delta}
\end{prooftree} &
\begin{prooftree}
    \hypo{\Gamma, A \seqlk B, \Delta}
    \infer1[\rulename{$\to$ R}]{\Gamma \seqlk A \to B, \Delta}
\end{prooftree} \\[4mm]
    \end{tabular}
    
\end{center}

\begin{definition}
A {\em derivation} in LK is a finite labelled tree whose internal nodes are labelled with sequents s.t. the label of each non-leaf node follows from the labels of its children using the calculus rules and leafs are axioms of the calculus. We say that a sequent ${A_1, \dots, A_n \seqlk B_1, \dots, B_m}$ is derivable in LK, if there is a derivation in LK with this sequent as the root label.
\end{definition}

 There are several known equivalent formulation of $LK$. The version we have presented is the one known as G3p. In this version, the standard structural rules in
 \figureword~\ref{fig:LK_structural_rules} are admissible (i.e. addition of these rules to the calculus does not change the set of sequents that can be derived). This makes the resulting derivations particularly well behaved as: (1) they can be constructed by safely applying the calculus rules bottom up, and (2) they satisfy the subformula property: all formulae occurring in a derivation are subformulae of the sequent to be proved. 
 

\begin{figure*}
    \centering
    \begin{tabular}{ccc}
\begin{prooftree}
    \hypo{\Gamma \seqlk \Delta}
    \infer1[\rulename{\weaklabelLK}]{\Gamma, A \seqlk \Delta}
\end{prooftree} &
\begin{prooftree}
    \hypo{\Gamma, A, A \seqlk \Delta}
    \infer1[\rulename{\contractlabelLK}]{\Gamma, A \seqlk \Delta}
\end{prooftree} &
\begin{prooftree}
    \hypo{\Gamma \seqlk A, \Delta}
    \hypo{\Gamma', A \seqlk \Delta'}
    \infer2[\rulename{\cutlabelLK}]{\Gamma, \Gamma' \seqlk \Delta, \Delta'}
\end{prooftree}
    \end{tabular}
    \caption{Structural rules of LK and cut (all admissible in LK). $\Gamma$, $\Gamma'$, $\Delta$, and $\Delta'$ denote an arbitrary multiset of formulae}
    \label{fig:LK_structural_rules}
\end{figure*}



\begin{theorem}[Gentzen]
A sequent ${A_1, \dots, A_n \seqlk B_1, \dots, B_m}$ is derivable in LK iff $A_1, \dots A_n \seqcl B_1 \vee \dots \vee B_m$ in classical logic.
\end{theorem}


\bibliographystyleap{kr}
\bibliographyap{biblio.bib}

\newpage

\section{Detailed Proofs}
\label{sec:ap_proofs}

\begin{lemman}[\ref{lem:struct_admissibility_2}]
\label{proof:struct_admissibility_2}
Rules \rulename{\weaklabel}, \rulename{\contractlabel} and \rulename{\cutlabel} are admissible for the calculus $\iocalc{2}$.
\end{lemman}
\begin{proof}We will use the characterization from \lemmaword~\ref{lem:char_2} to prove the admissibility for all three rules by reducing them to the admissibility of the corresponding rules in LK.

\vskip5mm
\textbf{Weakening.} Let ${G = \{(A_1, X_1), \dots, (A_n, X_n)\}}$. Denote $(A_{n+1}, X_{n+1}) = (A, X)$. Suppose $\ioseq{G}{B}{Y}$ is derivable. Then for every partition $(I,J) \in \partset(\{1,\dots,n\})$, \[ \textit{either } {B \seqlk \{A_i\}_{i \in I}} \textit{ or } {\{X_j\}_{j \in J} \seqlk Y} \textit{ is derivable in LK.} \] We can weaken the left sequent by adding the formula $A_{n+1}$ on the right:
\[ \textit{either } {B \seqlk \{A_i\}_{i \in I}, A_{n+1}} \textit{ or } {\{X_j\}_{j \in J} \seqlk Y} \textit{ is derivable in LK.} \]
Analogously, we can weaken the right sequent by adding the formula $X_{n+1}$ on the left:
\[ \textit{either } {B \seqlk \{A_i\}_{i \in I}} \textit{ or } {\{X_j\}_{j \in J}, X_{n+1} \seqlk Y} \textit{ is derivable in LK.} \]
Notice that these two statements give exactly the characterization of derivability of $\ioseq{(A, B), G}{B}{Y}$ by \lemmaword~\ref{lem:char_2} (first one for the partitions that have index $(n+1)$ in $I$, and second for the partitions that have index $(n+1)$ in $J$).

\vskip5mm
\textbf{Contraction.} Let ${G = \{(A_1, X_1), \dots, (A_n, X_n)\}}$. Suppose $\ioseq{(A, X), (A, X), G}{B}{Y}$ is derivable. Let's denote $(A_{n+1}, X_{n+1}) = (A_{n+2}, X_{n+2}) = (A,X)$ to have all pairs indexed for characterisation. 

Consider the partitions from ${\partset(\{1,\dots,n+2\})}$ that have the form ${(I \cup \{n+1,n+2\},J)}$. For any such partition, we have
\[ \textit{either } {B \seqlk \{A_i\}_{i \in I}, A, A} \textit{ or } {\{X_j\}_{j \in J} \seqlk Y} \textit{ is derivable in LK.} \]

After rewriting the first sequent into ${B \seqlk \{A_i\}_{i \in I}, A}$ using the classical contraction this condition turns exactly into the condition of the characterization for $\ioseq{(A,X), G}{B}{Y}$ (which has $(n+1)$ pairs on the left) for the partition $(I \cup \{n+1\}, J)$.

Now consider the partitions from ${\partset(\{1,\dots,n+2\})}$ that have the form ${(I,J \cup \{n+1,n+2\})}$. For any such partition, we have
\[ \textit{either } {B \seqlk \{A_i\}_{i \in I}} \textit{ or } {\{X_j\}_{j \in J}, X, X \seqlk Y} \textit{ is derivable in LK.} \]

After rewriting the second sequent into ${\{X_j\}_{j \in J}, X \seqlk Y}$ using classical contraction this condition turns exactly into the condition of the characterization for $\ioseq{(A,X), G}{B}{Y}$ for partition $(I, J \cup \{n+1\})$.

Notice that these two cases cover all partitions from  ${\partset(\{1,\dots,n+1\})}$, so $\ioseq{(A,X), G}{B}{Y}$ is derivable.

\vskip5mm
\textbf{Cut.} Let $G = {\{(D_1, W_1) \scomma \dots \scomma (D_m, W_m) \}}$ and $G' = {\{(A_1, X_1) \scomma \dots \scomma (A_n, X_n) \}}$ and denote $(A_{n+1}, X_{n+1}) = (C, Z)$. For derivability of $\ioseq{G, G'}{B}{Y}$ we need to show that for any partition $(I_1, J_1) \in {\partset(\{1,\dots,n\})}$ and any partition $(I_2, J_2) \in {\partset(\{1,\dots,m\})}$ the following holds:
\[ \textcolor{red}{(\ast)}\quad \textit{either } {B \seqlk \{A_i\}_{i \in I_1}, \{D_i\}_{i \in I_2}} \textit{ or } {\{X_j\}_{j \in J_1}, \{W_j\}_{j \in J_2} \seqlk Y} \textit{ is derivable in LK.} \]

We first use that $\ioseq{(C,Z), G}{B}{Y}$ is derivable.
First, we use characterization lemma for the partition ${(I_1 \cup \{n+1\}, J_1)}$ (so from the additional pair $(C,Z)$, $C$ goes to the inputs):
\[ \textcolor{blue}{(1)}\quad \textit{either } {B \seqlk \{A_i\}_{i \in I_1}, C} \textit{ or } {\{X_j\}_{j \in J_1} \seqlk Y} \textit{ is derivable in LK.} \]

Second, we use the characterization lemma for the partition ${(I_1, J_1 \cup \{n+1\})}$ (so from the additional pair $(C,Z)$, $Z$ goes to the outputs):
\[ \textcolor{blue}{(2)}\quad \textit{either } {B \seqlk \{A_i\}_{i \in I_1}} \textit{ or } {\{X_j\}_{j \in J_1}, Z \seqlk Y} \textit{ is derivable in LK.} \]

Now we use the derivability of $\ioseq{G'}{C}{Z}$. We use the characterization lemma for the partition $(I_2, J_2)$:
\[ \textcolor{blue}{(3)}\quad \textit{either } {C \seqlk \{D_i\}_{i \in I_2}} \textit{ or } {\{W_j\}_{j \in J_2} \seqlk Z} \textit{ is derivable in LK.} \]

Now we show that for any possible choice of derivable sequent (left or right) in $\textcolor{blue}{(1)}$, $\textcolor{blue}{(2)}$ and $\textcolor{blue}{(3)}$, we can always derive one of the sequents of $\textcolor{red}{(\ast)}$ (either by weakening or cut in LK):
\begin{itemize}
\item If the right sequent of $\textcolor{blue}{(1)}$ is derivable, then the right sequent of $\textcolor{red}{(\ast)}$ is derivable by weakening in LK.
\item If the left sequent of $\textcolor{blue}{(2)}$ is derivable, then the left sequent of $\textcolor{red}{(\ast)}$ is derivable by weakening in LK.
\item If the left sequent of $\textcolor{blue}{(3)}$ is derivable and the left sequent of $\textcolor{blue}{(1)}$ is derivable, we can use the LK cut on them and get exactly the left sequent of $\textcolor{red}{(\ast)}$.
\item If the right sequent of $\textcolor{blue}{(3)}$ is derivable and the right sequent of $\textcolor{blue}{(2)}$ is derivable, we can cut them using the LK cut and get exactly the right sequent of $\textcolor{red}{(\ast)}$.
\end{itemize}

Thus, we have proved $\textcolor{red}{(\ast)}$ for all possible cases.
\end{proof}

\begin{thn}[\ref{th:sound_comp_2} (Soundness and completeness of \iocalc{2})]
\label{proof:sound_comp_2}
$\ioseq{G}{B}{Y}$ is derivable in \iocalc{2} iff the pair $(B, Y)$ is derivable from the pairs in $G$ in \ioslogic{2}.
\end{thn}
\begin{proof}$ $\newline
\noindent\textbf{Completeness.} Proven constructively by induction on I/O derivation in \ioslogic{2}.

Suppose we have a derivation of $(B,Y)$ from pairs $G$ in \ioslogic{2}. We prove by induction on this derivation that for each occurring pair $(A, X)$, the derivability statement $\ioseq{G}{A}{X}$ is derivable in \iocalc{2}.

For the base case when the pair $(A, X)$ belongs to $G$ (and so $G = \{(A, X)\} \cup G'$ for some $G'$) we have the following derivation.

\[ \begin{prooftree}
    \hypo{A \wedge \neg A \seqlk}
    \infer1[\rulename{\inaxlabel}]{\ioseq{G'}{A \wedge \neg A}{X}}
    \hypo{\seqlk X \vee \neg X}
    \infer1[\rulename{\outaxlabel}]{\ioseq{G'}{A}{X \vee \neg X}}
    \infer2[\rulename{\pairelimlabel{2}}]{\ioseq{(A,X), G'}{A}{X}}
\end{prooftree} \]

For I/O rules with zero premises (\rulename{TOP} and \rulename{BOT}), we have the following derivations.

\[ \begin{prooftree}
    \hypo{\seqlk \top}
    \infer1[\rulename{\outaxlabel}]{\ioseq{G}{\top}{\top}}
\end{prooftree} \textit{\quad and \quad} \begin{prooftree}
    \hypo{\bot \seqlk}
    \infer1[\rulename{\inaxlabel}]{\ioseq{G}{\bot}{\bot}}
\end{prooftree}  \]

For I/O rules with one premise and a side condition of classical entailment (the rules \rulename{WO} and \rulename{SI}) we take the following derivations that apply cut to the inductive hypothesis. Notice the LK sequents in the derivations marked by the blue color. These sequents are derivable because they are equivalent to the side condition of a considered I/O rule ($X \seqcl Y$ for \rulename{WO}, and $B \seqcl A$ for \rulename{SI}).

For \rulename{WO}:

\[ \begin{prooftree}
    \hypo{\substack{\textit{by I.H.} \\ \vdots}}
    \infer1{\ioseq{G}{B}{X}}
    \hypo{B \wedge \neg B \seqlk}
    \infer1[\rulename{\inaxlabel}]{\ioseq{}{B \wedge \neg B}{Y}}
    \hypo{\textcolor{blue}{\seqlk Y \vee \neg X}}
    \infer1[\rulename{\outaxlabel}]{\ioseq{}{B}{Y \vee \neg X}}
    \infer2[\rulename{\pairelimlabel{2}}]{\ioseq{(B,X)}{B}{Y}}
    \infer2[\rulename{Cut}]{\ioseq{G}{B}{Y}}
\end{prooftree} \]

For \rulename{SI}:
\[ \begin{prooftree}
    \hypo{\substack{\textit{by I.H.} \\ \vdots}}
    \infer1{\ioseq{G}{A}{Y}}
    \hypo{\textcolor{blue}{B \wedge \neg A \seqlk}}
    \infer1[\rulename{\inaxlabel}]{\ioseq{}{B \wedge \neg A}{Y}}
    \hypo{\seqlk Y \vee \neg Y}
    \infer1[\rulename{\outaxlabel}]{\ioseq{}{B}{Y \vee \neg Y}}
    \infer2[\rulename{\pairelimlabel{2}}]{\ioseq{(A,Y)}{B}{Y}}
    \infer2[\rulename{Cut}]{\ioseq{G}{B}{Y}}
\end{prooftree}  \]

For the rules with two premises (\rulename{AND} and \rulename{OR}) we use the following derivations.

For \rulename{AND}:

\[ \hspace{-10mm} \begin{prooftree}
    \hypo{B \wedge \neg B \seqlk}
    \infer1[\rulename{\inaxlabel}]{\ioseq{(B, X_2)}{B \wedge \neg B}{X_1 \wedge X_2}}
    \hypo{B \wedge \neg B \seqlk}
    \infer1[\rulename{\inaxlabel}]{\ioseq{}{B \wedge \neg B}{(X_1 \wedge X_2) \vee \neg X_1}}
    \hypo{\seqlk (X_1 \wedge X_2) \vee \neg X_1 \vee \neg X_2}
    \infer1[\rulename{\outaxlabel}]{\ioseq{}{B}{(X_1 \wedge X_2) \vee \neg X_1 \vee \neg X_2}}
    \infer2[\rulename{\pairelimlabel{2}}]{\ioseq{(B, X_2)}{B}{(X_1 \wedge X_2) \vee \neg X_1}}
    \infer2[\rulename{\pairelimlabel{2}}]{\ioseq{(B, X_1) \scomma (B, X_2)}{B}{X_1 \land X_2}}
\end{prooftree}  \]

For \rulename{OR}:

\[ \hspace{-10mm} \begin{prooftree}
    \hypo{(A_1 \vee A_2) \wedge \neg A_1 \wedge \neg A_2 \seqlk}
    \infer1[\rulename{\inaxlabel}]{\ioseq{}{(A_1 \vee A_2) \wedge \neg A_1 \wedge \neg A_2}{Y}}
    \hypo{\seqlk Y \vee \neg Y}
    \infer1[\rulename{\outaxlabel}]{\ioseq{}{(A_1 \vee A_2) \wedge \neg A_1}{Y \vee \neg Y}}
    \infer2[\rulename{\pairelimlabel{2}}]{\ioseq{(A_2, Y)}{(A_1 \vee A_2) \wedge \neg A_1}{Y}}
    \hypo{\seqlk Y \vee \neg Y}
    \infer1[\rulename{\outaxlabel}]{\ioseq{(A_2, Y)}{A_1 \vee A_2}{Y \vee \neg Y}}
    \infer2[\rulename{\pairelimlabel{2}}]{\ioseq{(A_1, Y) \scomma (A_2, Y)}{A_1 \vee A_2}{Y}}
\end{prooftree}  \]

We afterwards cut these derivations twice with both inductive hypotheses and use contraction to remove duplicate pairs from the result. We show how to do it for \rulename{AND} (for \rulename{OR} it is exactly the same):

\[ \begin{prooftree}
    \hypo{\substack{\textit{by I.H.} \\ \vdots}}
    \infer1{\ioseq{G}{B}{X_1}}
    \hypo{\substack{\textit{by I.H.} \\ \vdots}}
    \infer1{\ioseq{G}{B}{X_2}}
    \hypo{\ioseq{(B, X_1) \scomma (B, X_2)}{B}{X_1 \land X_2}}
    \infer2[\rulename{\cutlabel}]{\ioseq{G \scomma (B, X_1)}{B}{X_1 \land X_2}}
    \infer2[\rulename{\cutlabel}]{\ioseq{G, G}{B}{X_1 \land X_2}}
    \infer1[\rulename{\contractlabel} $\times n$]{\ioseq{G}{B}{X_1 \land X_2}}
\end{prooftree}  \]

\textbf{Soundness.}

Proven constructively by induction on derivations in \iocalc{2}.

Inductive base: the rules \rulename{\inaxlabel} and \rulename{\outaxlabel}.

If $B \seqlk$ (which means $B \seqcl \bot$ in classical logic) we have the following I/O derivation for the required pair $(B,Y)$ in \ioslogic{2}.

\[ \begin{prooftree}
    \hypo{}
    \infer1[\rulename{BOT}]{(\bot,\bot)}
    \infer1[\rulename{SI}]{(B, \bot)}
    \infer1[\rulename{WO}]{(B, Y)}
\end{prooftree}  \]

If $\seqlk Y$ (which means $\top \seqcl Y$ in classical logic) we have the following I/O derivation for the required pair $(B, Y)$ in \ioslogic{2}.

\[ \begin{prooftree}
    \hypo{}
    \infer1[\rulename{TOP}]{(\top,\top)}
    \infer1[\rulename{WO}]{(\top, Y)}
    \infer1[\rulename{SI}]{(B, Y)}
\end{prooftree}  \]

The inductive step is the rule $\rulename{\pairelimlabel{2}}$. By inductive hypotheses we have that both pairs ${(B \wedge \neg A, Y)}$ and ${(B, Y \vee \neg X)}$ are derivable from pairs $G$ in \ioslogic{2}. We now need to prove that ${(B,Y)}$ is derivable from the pairs in ${G \cup \{(A,X)\}}$. We do it with the following I/O derivation of the pair $(B,Y)$ from the pairs $(A,X)$, $(B \wedge \neg A,Y)$ and $(B, Y \vee \neg X)$ in the logic \ioslogic{2}. If you take it and plug in the derivations of $(B \wedge \neg A,Y)$ and $(B, Y \vee \neg X)$ from $G$ obtained by the inductive hypotheses you will get the required derivation from ${G \cup \{(A,X)\}}$.

\[ \begin{prooftree}
    \hypo{(A, X)}
    \infer1[\rulename{WO}]{(A, Y \vee X)}
    \infer1[\rulename{SI}]{(B \wedge A, Y \vee X)}
    \hypo{(B, Y \vee \neg X)}
    \infer1[\rulename{SI}]{(B \wedge A, Y \vee \neg X)}
    \infer2[\rulename{AND}]{(B \wedge A, (Y \vee X) \wedge (Y \vee \neg X))}
    \infer1[\rulename{WO}]{(B \wedge A, Y \vee (X \wedge \neg X))}
    \infer1[\rulename{WO}]{(B \wedge A, Y)}
    \hypo{(B \wedge \neg A, Y)}
    \infer2[\rulename{OR}]{((B \wedge A) \vee (B \wedge \neg A), Y)}
    \infer1[\rulename{SI}]{(B \wedge (A \vee \neg A), Y)}
    \infer1[\rulename{SI}]{(B, Y)}
\end{prooftree}  \]
\end{proof}

\begin{lemman}[\ref{lem:char_4}]
\label{proof:char_4}
$\ioseq{(A_1, X_1), \dots, (A_n, X_n)}{B}{Y}$ is derivable in \iocalc{4} iff for all partitions $(I, J) \in \partset(\{1, \dots , n\})$, either ${B, \{X_j\}_{j \in J} \seqlk
\{A_i\}_{i \in I}}$ or ${\{X_j\}_{j \in J} \seqlk Y}$ is derivable in LK.
\end{lemman}
\begin{proof}
Analogous to the proof of \lemmaword~\ref{lem:char_2}. The only difference is that the sequents that we have in the conditions for derivability obtained from the inductive hypothesis applied to premises in the inductive step are different (due to modifications in both pair elimination rule and characterization statement). Specifically, the existence of the derivation that starts with the application of the rule \rulename{\pairelimlabel{4}} that eliminates a pair $(A_k, X_k)$, is equivalent (by applying inductive hypothesis to the premises) to the derivability of the following sequents for all $(I', J') \in \partset(\{1, \dots, n+1\} \setminus \{k\})$:

\begin{itemize}
    \item  (a1) $B \wedge \neg A_k, \{X_j\}_{j \in J'} \seqlk \{A_i\}_{i \in I'}$ or (a2) $\{X_j\}_{j \in J'} \seqlk Y$
    
    \item (b1) $B \wedge X_k, \{X_j\}_{j \in J'} \seqlk \{A_i\}_{i \in I'}$ or (b2) $\{X_j\}_{j \in J'} \seqlk Y \vee \neg X_k$
\end{itemize}

But these two modifications complement each other: they both introduce deriving outputs on the left side of the first sequent. Therefore when moving $A_k$ and $X_k$ to the other deriving inputs/outputs, we are also able to move $X_k$ from $B \wedge X_k$ to $\{X_j\}_{j \in J'}$ on the same side. This way we again get two equivalent conditions for every $(I', J') \in \partset(\{1, \dots, n+1\} \setminus \{k\})$:

\begin{itemize}
    \item  (a1) $B, \{X_j\}_{j \in J'} \seqlk \{A_i\}_{i \in I' \cup \{k\}}$ or (a2) $\{X_j\}_{j \in J'} \seqlk Y$
    
    \item (b1) $B, \{X_j\}_{j \in J' \cup \{k\}} \seqlk \{A_i\}_{i \in I'}$ or (b2) $\{X_j\}_{j \in J'  \cup \{k\}} \seqlk Y$
\end{itemize}

Which gives exactly the statement of the lemma.
\end{proof}

\begin{lemma}
\label{lem:calc_subsums_4_2}
If $\ioseq{G}{B}{Y}$ is derivable in \iocalc{2} it is also derivable in \iocalc{4}.
\end{lemma}
\begin{proof}
We prove the more general statement: if $\ioseq{G}{B}{Y}$ is derivable in \iocalc{2} then for any $B'$ such that $B' \seqcl B$ in classical logic, $\ioseq{G}{B'}{Y}$ is also derivable in \iocalc{4}.

The proof proceeds by induction on the derivation of $\ioseq{G}{B}{Y}$ in \iocalc{2}.

\begin{itemize}
\item If $\ioseq{G}{B}{Y}$ is derived by the rule \rulename{\inaxlabel} then $B \seqlk$ is derivable in LK, so $B' \seqlk$ is also derivable in LK and $\ioseq{G}{B'}{Y}$ can be derived in \iologic{4} by \rulename{\inaxlabel}.
\item If $\ioseq{G}{B}{Y}$ is derived by the rule \rulename{\outaxlabel} then $\seqlk Y$ is derivable in LK, so $\ioseq{G}{B'}{Y}$ can be derived in \iocalc{4} by \rulename{\outaxlabel}.
\item Suppose $\ioseq{G}{B}{Y}$ is derived by the rule \rulename{\pairelimlabel{2}} from the premises $\ioseq{G'}{B \wedge \neg A}{Y}$ and $\ioseq{G'}{B}{Y \vee \neg X}$ for some pair $(A,X)$ such that $G = G' \cup (A,X)$. By the inductive hypotheses $\ioseq{G'}{B' \wedge \neg A}{Y}$ and $\ioseq{G'}{B' \wedge X}{Y \vee \neg X}$ are derivable in \iocalc{4} (since from $B' \seqcl B$ follows that $B' \wedge \neg A \seqcl B \wedge \neg A$ and $B' \wedge X \seqcl B$ classically). We then can derive $\ioseq{G}{B'}{Y}$ in \iocalc{4} from these premises by \rulename{\pairelimlabel{4}}.
\end{itemize}
\end{proof}

\begin{lemma}
\label{lem:struct_admissibility_4}
The rules \rulename{\weaklabel}, \rulename{\contractlabel} and \rulename{\cutlabel} are admissible in the calculus $\iocalc{4}$.
\end{lemma}
\begin{proof}
Analogous to the proof \ref{proof:struct_admissibility_2} of \lemmaword~\ref{lem:struct_admissibility_2}. 
By the characterization lemma (\lemmaword~\ref{lem:char_4}) the difference is that the output part $\{X_j\}_{j \in J}$ appear on the left of the first sequent. This does not affect the proof of admissibility of \rulename{\weaklabel} and \rulename{\contractlabel} --- weakening and contraction from LK can be applied on the left where the additional formulae appear the same way as in other parts to reduce the sequents to the required form.

For \rulename{\cutlabel} rule admissibility there is one additional LK cut that may be required in one of the cases. Specifically to prove admissibility of \rulename{\cutlabel} following the proof \ref{proof:struct_admissibility_2} of Lemma~\ref{lem:struct_admissibility_2} we have to prove the following characterization statement for any partition $(I_1, J_1) \in \partset(\{1,\dots,n\})$ and partition $(I_2, J_2) \in \partset(\{1,\dots,m\})$. 

\[ \textcolor{red}{(\ast)}\quad \textit{either } {B, \{X_j\}_{j \in J_1}, \{W_j\}_{j \in J_2} \seqlk \{A_i\}_{i \in I_1}, \{D_i\}_{i \in I_2}} \textit{ or } {\{X_j\}_{j \in J_1}, \{W_j\}_{j \in J_2} \seqlk Y} \textit{ is derivable in LK.} \]

Again, we first get the following statement by applying characterization lemma to the derivation of $\ioseq{(C,Z), G}{B}{Y}$ for the partition $(I_1 \cup \{n+1\}, J_1)$:
\[ \textcolor{blue}{(1)}\quad \textit{either } {B, \{X_j\}_{j \in J_1} \seqlk \{A_i\}_{i \in I_1}, C} \textit{ or } {\{X_j\}_{j \in J_1} \seqlk Y} \textit{ is derivable in LK.} \]


Second, we get the following statement by applying the characterization lemma to the derivation of $\ioseq{(C,Z), G}{B}{Y}$ for the partition $(I_1, J_1 \cup \{n+1\})$:
\[ \textcolor{blue}{(2)}\quad \textit{either } {B, Z, \{X_j\}_{j \in J_1} \seqlk \{A_i\}_{i \in I_1}} \textit{ or } {Z, \{X_j\}_{j \in J_1} \seqlk Y} \textit{ is derivable in LK.} \]


Third, we get the following statement by applying the characterization lemma to the derivation of $\ioseq{G'}{C}{Z}$ for the partition $(I_2, J_2)$:
\[ \textcolor{blue}{(3)}\quad \textit{either } {C, \{W_j\}_{j \in J_2} \seqlk \{D_i\}_{i \in I_2}} \textit{ or } {\{W_j\}_{j \in J_2} \seqlk Z} \textit{ is derivable in LK.} \]

Now we show that for any possible choice of derivable sequent (left or right) in $\textcolor{blue}{(1)}$, $\textcolor{blue}{(2)}$ and $\textcolor{blue}{(3)}$, we can always derive one of the sequents of $\textcolor{red}{(\ast)}$ (either by weakening or by cut in LK):
\begin{itemize}
\item If the right sequent of $\textcolor{blue}{(1)}$ is derivable, then the right sequent of $\textcolor{red}{(\ast)}$ is derivable by weakening in LK.
\item If the left sequent of $\textcolor{blue}{(1)}$ is derivable and the left sequent of $\textcolor{blue}{(3)}$ is derivable, we can cut them using the LK cut and get exactly the left sequent of $\textcolor{red}{(\ast)}$.
\item If the right sequent of $\textcolor{blue}{(3)}$ is derivable and the left sequent of $\textcolor{blue}{(2)}$ is derivable, we can cut them using the LK cut and weaken the result to get the left sequent of $\textcolor{red}{(\ast)}$.
\item If the right sequent of $\textcolor{blue}{(3)}$ is derivable and the right sequent of $\textcolor{blue}{(2)}$ is derivable, we can cut them using the LK cut and get exactly the right sequent of $\textcolor{red}{(\ast)}$.
\end{itemize}

Thus, we have proved $\textcolor{red}{(\ast)}$ for all possible cases.
\end{proof}

\begin{thn}[\ref{th:sound_comp_4}]
\label{proof:sound_comp_4}
$\ioseq{G}{B}{Y}$ is derivable in \iocalc{4} iff  $(B, Y)$ is derivable from the pairs in $G$ in \ioslogic{4}.
\end{thn}
\begin{proof}
Analogous to the proof~\ref{proof:sound_comp_2} of the \theoremword~\ref{th:sound_comp_2}.

\noindent\textbf{Completeness.}

We need to check the derivability of all rules of \ioslogic{4} in \iocalc{4}. While this can be done directly, we also can use the derivability of all axioms \ioslogic{2} in \iocalc{2}, that was proven in the proof~\ref{proof:sound_comp_2} of \lemmaword~\ref{lem:sound_comp_2}, and the additional simple fact that anything derivable in \iocalc{2} is also derivable in \iocalc{4} (lemma~\ref{lem:calc_subsums_4_2}).

Then the only rule of \ioslogic{4} left to check is the peculiar rule \rulename{CT}. Its derivation in \iocalc{4} is given below.
%
%

\[ \begin{prooftree}
    \hypo{A \wedge \neg (A \wedge X) \wedge \neg A \seqlk}
    \infer1[\rulename{\inaxlabel}]{\ioseq{}{A \wedge \neg (A \wedge X) \wedge \neg A}{Y}}
    \hypo{A \wedge \neg (A \wedge X) \wedge X \seqlk}
    \infer1[\rulename{\inaxlabel}]{\ioseq{}{A \wedge \neg (A \wedge X) \wedge X}{Y \vee \neg X}}
    \infer2[\rulename{\pairelimlabel{4}}]{\ioseq{(A,X)}{A \wedge \neg (A \wedge X)}{Y}}
    \hypo{\seqlk Y \vee \neg Y}
    \infer1[\rulename{\outaxlabel}]{\ioseq{(A, X)}{A \wedge Y}{Y \vee \neg Y}}
    \infer2[\rulename{\pairelimlabel{4}}]{\ioseq{(A, X) \scomma (A \wedge X, Y)}{A}{Y}}
    \end{prooftree} \]

Then this derivations of \ioslogic{4} rules are combined using structural rules to build a derivation in \iocalc{4} analogously to the proof~\ref{proof:sound_comp_2} of \theoremword~\ref{th:sound_comp_2}. Admissibility of structural rules in \iocalc{4} is given by \lemmaword~\ref{lem:struct_admissibility_4}.

\textbf{Soundness.}

Analogous to the proof~\ref{proof:sound_comp_2} of \theoremword~\ref{th:sound_comp_2}. The soundness of the rule \rulename{\pairelimlabel{4}} in \iocalc{4} is given by the following derivation.

\[ \begin{prooftree}
    \hypo{(A, X)}
    \infer1[\rulename{WO}]{(A, Y \vee X)}
    \infer1[\rulename{SI}]{(B \wedge A, Y \vee X)}
        \hypo{(A, X)}
    \infer1[\rulename{SI}]{(B \wedge A, X)}
    \hypo{(B \wedge X, Y \vee \neg X)}
    \infer1[\rulename{SI}]{(B \wedge A \wedge X, Y \vee \neg X)}
    \infer2[\rulename{CT}]{(B \wedge A, Y \vee \neg X)}
    \infer2[\rulename{AND}]{(B \wedge A, (Y \vee X) \wedge (Y \vee \neg X))}
    \infer1[\rulename{WO}]{(B \wedge A, Y \vee (X \wedge \neg X))}
    \infer1[\rulename{WO}]{(B \wedge A, Y)}
    \hypo{(B \wedge \neg A, Y)}
    \infer2[\rulename{OR}]{((B \wedge A) \vee (B \wedge \neg A), Y)}
    \infer1[\rulename{SI}]{(B \wedge (A \vee \neg A), Y)}
    \infer1[\rulename{SI}]{(B, Y)}
\end{prooftree}  \]
\end{proof}

\begin{lemman}[\ref{lem:semi_invert_3}]
\label{proof:semi_invert_3}
If $\ioseq{(A,X), G}{B}{Y}$ is derivable in \iocalc{3}, then $\ioseq{G}{B \wedge X}{Y \vee \neg X}$ is also derivable \iocalc{3}.
\end{lemman}
\begin{proof}
By induction on the derivation of $\ioseq{(A,X), G}{B}{Y}$ in \iocalc{3}.
\begin{itemize}
\item Suppose $\ioseq{(A,X), G}{B}{Y}$ is derived by the rule \rulename{\inaxlabel}. Then $B \seqlk$ is derivable in LK, so $B \wedge X \seqlk$ is also derivable in LK and $\ioseq{G}{B \wedge X}{Y \vee \neg X}$ can be also derived by the rule \rulename{\inaxlabel}.
\item Suppose $\ioseq{(A,X), G}{B}{Y}$ is derived by the rule \rulename{\outaxlabel}. Then $\seqlk Y$ is derivable in LK, so $\seqlk Y \vee \neg X$ is also derivable in LK and $\ioseq{G}{B \wedge X}{Y \vee \neg X}$ can be also derived by the rule \rulename{\outaxlabel}.
\item If $\ioseq{(A,X), G}{B}{Y}$ is derived by elimination of the pair $(A,X)$ with the rule \rulename{\pairelimlabel{3}}, than there is a derivation of the premise $\ioseq{G}{B \wedge X}{Y \vee \neg X}$.
\item Suppose $\ioseq{(A,X), G}{B}{Y}$ is derived by elimination of some other pair $(A',X')$ with the rule \rulename{\pairelimlabel{3}} ($G = \{(A',X')\} \cup G'$ for some $G'$). The first premise is then $B \seqlk A'$, which infers $B \wedge X \seqcl A'$. The second premise is $\ioseq{(A,X), G'}{B \wedge X'}{Y \vee \neg X'}$ to which we can apply the inductive hypothesis to get a derivation of $\ioseq{G'}{B \wedge X' \wedge X}{Y \vee \neg X' \vee \neg X}$. The required sequent can be derived as follows:

\[ \begin{prooftree}
    \hypo{B \wedge X \seqlk A'}
    \hypo{\ioseq{G'}{B \wedge X' \wedge X}{Y \vee \neg X' \vee \neg X}}
    \infer2[\rulename{\pairelimlabel{3}}]{\ioseq{(A', X'), G'}{B \wedge X}{Y \vee \neg X}}
\end{prooftree}  \]
\end{itemize}
\end{proof}

\begin{lemma}
\label{lem:struct_admissibility_1_3}
The rules \rulename{\weaklabel}, \rulename{\contractlabel} and \rulename{\cutlabel} are admissible for the calculi $\iocalc{1}$ and $\iocalc{3}$.
\end{lemma}
\begin{proof}
Analogous to the proof of Lemma~\ref{lem:struct_admissibility_2} for the calculus \iocalc{2} and its analogue for the calculus \iocalc{4}. The difference is that the characterization lemma in the case of \iocalc{1} and \iocalc{3} does not have only two, but $(|I| + 2)$ LK-sequents one of which should be derivable for each partition $(I,J)$. However the form of these sequents is the same, so the structure of the proof does not change.

For admissibility of weakening and contraction the modification is minimal: instead of applying the LK rules of weakening and contraction to add/duplicate the formula $A$ (for weakened/contracted pair $(A,X)$) we add/duplicate the whole separate sequent that has $A$ on the right. Since the characterization statement requires at least one of the sequents to be derivable, addition of a new sequent or duplication of the existing sequent results in a statement that is weaker (or equivalent), as needed.

We will now show how LK cuts (and weakenings) should be applied to establish cut admissibility in \iocalc{3} (case of \iocalc{1} is exactly the same, but the sequents are simpler, since they do not have additional outputs on the left).

Using the characterization lemma~\ref{lem:char_3}, for arbitrary partitions $(I_1 ,J_1) \in \partset(\{1,\dots,n\})$ and $(I_2 ,J_2) \in \partset(\{1,\dots,m\})$ we need to prove the following characterization statement

\textcolor{red}{$(\ast)$}\quad \textit{at least one of the following $(|I_1| + |I_2| + 2)$ sequents should be derivable:
\begin{itemize}
\item[\textcolor{red}{$(\ast1)$}] $B, \{X_j\}_{j \in J_1}, \{W_j\}_{j \in J_2} \seqlk A_i$ for some $i \in I_1$
\item[\textcolor{red}{$(\ast2)$}] $B, \{X_j\}_{j \in J_1}, \{W_j\}_{j \in J_2} \seqlk D_l$ for some $l \in I_2$
\item[\textcolor{red}{$(\ast3)$}] $B, \{X_j\}_{j \in J_1}, \{W_j\}_{j \in J_2} \seqlk$
\item[\textcolor{red}{$(\ast4)$}] $\{X_j\}_{j \in J_1}, \{W_j\}_{j \in J_2} \seqlk Y$
\end{itemize}}

Again, we first apply the characterization lemma to derivability of $\ioseq{(C,Z), G}{B}{Y}$ for the partition $(I_1 \cup \{n+1\}, J_1)$ and get the following statement

\textcolor{blue}{$(a)$}\quad \textit{at least one of the following $(|I_1| + 3)$ sequents should be derivable:
\begin{itemize}
\item[\textcolor{blue}{$(a1)$}] $B, \{X_j\}_{j \in J_1} \seqlk A_i$ for some $i \in I_1$
\item[\textcolor{blue}{$(a2)$}] $B, \{X_j\}_{j \in J_1} \seqlk C$
\item[\textcolor{blue}{$(a3)$}] $B, \{X_j\}_{j \in J_1} \seqlk$
\item[\textcolor{blue}{$(a4)$}] $\{X_j\}_{j \in J_1} \seqlk Y$
\end{itemize}}

Second, we get the following statement by applying the characterization lemma to derivability $\ioseq{(C,Z), G}{B}{Y}$ for the partition $(I_1, J_1 \cup \{n+1\})$:

\textcolor{blue}{$(b)$}\quad \textit{at least one of the following $(|I_1| + 2)$ sequents should be derivable:
\begin{itemize}
\item[\textcolor{blue}{$(b1)$}] $B, \{X_j\}_{j \in J_1}, Z \seqlk A_i$ for some $i \in I_1$
\item[\textcolor{blue}{$(b2)$}] $B, \{X_j\}_{j \in J_1}, Z \seqlk$
\item[\textcolor{blue}{$(b3)$}] $\{X_j\}_{j \in J_1}, Z \seqlk Y$
\end{itemize}}

Third, we get the following statement by applying the characterization lemma to derivability $\ioseq{G'}{C}{Z}$ for the partition $(I_2, J_2)$:

\textcolor{blue}{$(c)$}\quad \textit{at least one of the following $(|I_2| + 2)$ sequents should be derivable:
\begin{itemize}
\item[\textcolor{blue}{$(c1)$}] $C, \{W_j\}_{j \in J_2} \seqlk D_l$ for some $l \in I_2$
\item[\textcolor{blue}{$(c2)$}] $C, \{W_j\}_{j \in J_2} \seqlk$
\item[\textcolor{blue}{$(c3)$}] $\{W_j\}_{j \in J_2} \seqlk Z$
\end{itemize}}

Now we show that for any possible choice of derivable sequent in $\textcolor{blue}{(1)}$, $\textcolor{blue}{(2)}$ and $\textcolor{blue}{(3)}$, we can always derive one of the sequents of $\textcolor{red}{(\ast)}$ (either by weakening or cut in LK):
\begin{itemize}
\item If one of the sequents in $\textcolor{blue}{(a1)}$ is derivable, one of sequents $\textcolor{red}{(\ast1)}$ is derivable by weakening.

\item If the sequent $\textcolor{blue}{(a3)}$ is derivable, the sequent $\textcolor{red}{(\ast3)}$ is derivable by weakening.

\item If the sequent $\textcolor{blue}{(a4)}$ is derivable, the sequent $\textcolor{red}{(\ast4)}$ is derivable by weakening.

\item If the sequent $\textcolor{blue}{(a2)}$ and one of the sequents $\textcolor{blue}{(c1)}$ are derivable we can use the LK cut on them and get one of the sequents $\textcolor{red}{(\ast2)}$.

\item If the sequent $\textcolor{blue}{(a2)}$ and the sequent $\textcolor{blue}{(c2)}$ are derivable we apply the LK cut to them and get the sequent $\textcolor{red}{(\ast3)}$.

\item If the sequent $\textcolor{blue}{(c3)}$ and one of the sequents $\textcolor{blue}{(b1)}$ are derivable we use the LK cut and get one of the sequents $\textcolor{red}{(\ast1)}$.

\item If the sequent $\textcolor{blue}{(c3)}$ and the sequent $\textcolor{blue}{(b2)}$ are derivable we can use the LK cut  and get the sequent $\textcolor{red}{(\ast3)}$.

\item If the sequent $\textcolor{blue}{(c3)}$ and the sequent $\textcolor{blue}{(b3)}$ are derivable using the LK cut we get the sequent $\textcolor{red}{(\ast4)}$.
\end{itemize}

Thus, we have proved that one of the sequents $\textcolor{red}{(\ast)}$ are derivable for all possible cases.
\end{proof}

\begin{thn}[\ref{th:sound_comp_1_3}]
\label{proof:sound_comp_1_3}
For both $k \in \{1,3\}$, $\ioseq{G}{B}{Y}$ is derivable in \iocalc{k} iff the pair $(B, Y)$ is derivable from the pairs in $G$ in \ioslogic{k}.
\end{thn}
\begin{proof}
Analogous to the proofs of Theoremes \ref{th:sound_comp_2} and \ref{th:sound_comp_4} for the calculi \iocalc{2} and \iocalc{4}.

\textbf{Completeness}\\
We need to check that all I/O rules of \ioslogic{1} (plus adequacy of derivability $\ioseq{(A,X)}{A}{X}$) are derivable in \iocalc{1} and all I/O rules of \ioslogic{3} are derivable in \iocalc{3}; combining instances of these rules via (admissible) structural rules is the same as in the proof of completeness of \iocalc{2}.

Below are derivation of I/O rules of \ioslogic{1} in \iocalc{1}.

\begin{enumerate}
\item Base case ($(A,X)$ belongs to $G$)

\[ \begin{prooftree}
    \hypo{A \seqlk A}
    \hypo{\seqlk X \vee \neg X}
    \infer1[\rulename{\outaxlabel}]{\ioseq{G'}{A}{X \vee \neg X}}
    \infer2[\rulename{\pairelimlabel{1}}]{\ioseq{(A,X), G'}{A}{X}}
\end{prooftree} \]

\item \rulename{TOP}

\[ \begin{prooftree}
    \hypo{\seqlk \top}
    \infer1[\rulename{\outaxlabel}]{\ioseq{}{\top}{\top}}
\end{prooftree} \]

\item \rulename{BOT}

\[ \begin{prooftree}
    \hypo{\bot \seqlk}
    \infer1[\rulename{\inaxlabel}]{\ioseq{}{\bot}{\bot}}
\end{prooftree} \]

\item \rulename{WO} (the blue LK-sequent is derivable due to the side condition that $X \seqcl Y$ classically)

\[ \begin{prooftree}
    \hypo{B \seqlk B}
    \hypo{\textcolor{blue}{\seqlk Y \vee \neg X}}
    \infer1[\rulename{\outaxlabel}]{\ioseq{}{B}{Y \vee \neg X}}
    \infer2[\rulename{\pairelimlabel{1}}]{\ioseq{(B,X)}{B}{Y}}
\end{prooftree} \]

\item \rulename{SI} (the blue LK-sequent is derivable due to the side condition that $B \seqcl A$ classically)

\[ \begin{prooftree}
    \hypo{\textcolor{blue}{B \seqlk A}}
    \hypo{\seqlk Y \vee \neg Y}
    \infer1[\rulename{\outaxlabel}]{\ioseq{}{B}{Y \vee \neg Y}}
    \infer2[\rulename{\pairelimlabel{1}}]{\ioseq{(A,Y)}{B}{Y}}
\end{prooftree} \]

\item \rulename{AND}

\[ \begin{prooftree}
    \hypo{B \seqlk B}
    \hypo{B \seqlk B}
    \hypo{\seqlk (X_1 \wedge X_2) \vee \neg X_1 \vee \neg X_2}
    \infer1[\rulename{\outaxlabel}]{\ioseq{}{B}{(X_1 \wedge X_2) \vee \neg X_1 \vee \neg X_2}}
    \infer2[\rulename{\pairelimlabel{1}}]{\ioseq{(B, X_2)}{B}{(X_1 \wedge X_2) \vee \neg X_1}}
    \infer2[\rulename{\pairelimlabel{1}}]{\ioseq{(B, X_1) \scomma (B, X_2)}{B}{X_1 \land X_2}}
\end{prooftree}  \]

\end{enumerate}

These rules can be derived the same way for \iocalc{3}, or alternatively analogously to the \lemmaword~\ref{lem:calc_subsums_4_2} we can prove that every sequent derivable in \iocalc{1} is also derivable in \iocalc{3} and refer to the derivations in \iocalc{1} above. The only rule left to prove is the derivability in \iocalc{3} of the additional axiom \rulename{CT} of \ioslogic{3}. The derivation is below.

\[ \begin{prooftree}
    \hypo{A \seqlk A}
    \hypo{A \wedge X \seqlk A \wedge X}
    \hypo{\seqlk Y \vee \neg X \vee \neg Y}
    \infer1[\rulename{\outaxlabel}]{\ioseq{}{A \wedge X \wedge Y}{Y \vee \neg X \vee \neg Y}}
    \infer2[\rulename{\pairelimlabel{3}}]{\ioseq{(A \wedge X, Y)}{A \wedge X}{Y \vee \neg X}}
    \infer2[\rulename{\pairelimlabel{3}}]{\ioseq{(A, X) \scomma (A \wedge X, Y)}{A}{Y}}
\end{prooftree}  \]

Then the derivations of the rules are combined using structural rules to build a derivation in \iocalc{1} (or \iocalc{3}) analogously to the proof~\ref{proof:sound_comp_2} of \theoremword~\ref{th:sound_comp_2}. Admissibility of structural rules in \iocalc{4} is given by \lemmaword~\ref{lem:struct_admissibility_1_3}.

\textbf{Soundness}

Analogous to the proof~\ref{proof:sound_comp_2} of \theoremword~\ref{th:sound_comp_2}.

To establish soundness of the rule \rulename{\pairelimlabel{1}} in \ioslogic{1}, notice that if first premise $B \seqlk A$ is derivable, then $B \seqcl A$ (by soundness and completeness of LK), and therefore ${B \seqcl B \wedge A}$. Then the following derivation provides the soundness for the rule \rulename{\pairelimlabel{1}} in \ioslogic{1}.

\[ \begin{prooftree}
    \hypo{(A, X)}
    \infer1[\rulename{WO}]{(A, Y \vee X)}
    \infer1[\rulename{SI}]{(B \wedge A, Y \vee X)}
    \hypo{(B, Y \vee \neg X)}
    \infer1[\rulename{SI}]{(B \wedge A, Y \vee \neg X)}
    \infer2[\rulename{AND}]{(B \wedge A, (Y \vee X) \wedge (Y \vee \neg X))}
    \infer1[\rulename{WO}]{(B \wedge A, Y \vee (X \wedge \neg X))}
    \infer1[\rulename{WO}]{(B \wedge A, Y)}
    \infer1[\rulename{SI}]{(B, Y)}
\end{prooftree}  \]

Analogously, for the soundness of the rule \rulename{\pairelimlabel{3}} in \ioslogic{3} we have the following derivation.

\[ \begin{prooftree}
    \hypo{(A, X)}
    \infer1[\rulename{WO}]{(A, Y \vee X)}
    \infer1[\rulename{SI}]{(B \wedge A, Y \vee X)}
        \hypo{(A, X)}
    \infer1[\rulename{SI}]{(B \wedge A, X)}
    \hypo{(B \wedge X, Y \vee \neg X)}
    \infer1[\rulename{SI}]{(B \wedge A \wedge X, Y \vee \neg X)}
    \infer2[\rulename{CT}]{(B \wedge A, Y \vee \neg X)}
    \infer2[\rulename{AND}]{(B \wedge A, (Y \vee X) \wedge (Y \vee \neg X))}
    \infer1[\rulename{WO}]{(B \wedge A, Y \vee (X \wedge \neg X))}
    \infer1[\rulename{WO}]{(B \wedge A, Y)}
    \infer1[\rulename{SI}]{(B, Y)}
\end{prooftree}  \]

\end{proof}

\end{document}